\documentclass{llncs}

\usepackage[english]{babel}

\usepackage[%
rm={oldstyle=false,proportional=true},%
sf={oldstyle=false,proportional=true},%
tt={oldstyle=false,proportional=true,variable=true},%
qt=false%
]{cfr-lm}
\usepackage{listings}

\usepackage{graphicx}
\graphicspath{ {figures/} }

\usepackage{amsmath}
\usepackage{amssymb}

\usepackage{paralist}

\usepackage{cite}

\usepackage[T1]{fontenc}

\usepackage[math]{blindtext}

\usepackage{csquotes}

\usepackage{microtype}

\usepackage{url}
\makeatletter
\g@addto@macro{\UrlBreaks}{\UrlOrds}
\makeatother

\usepackage{xcolor}

\usepackage{amsmath}
\usepackage{amsfonts}

\usepackage[
bookmarks=false,
breaklinks=true,
colorlinks=true,
linkcolor=black,
citecolor=black,
urlcolor=black,
pdfpagelayout=SinglePage,
pdfstartview=Fit
]{hyperref}
\usepackage[all]{hypcap}

\usepackage{pdfcomment}

\newif\ifsubmit\submittrue

\newcommand{\stacklabel}[1]{\stackrel{\smash{\scriptscriptstyle \mathrm{#1}}}}
\newcommand{\defeq}{\stacklabel{def}=}
\newcommand{\nzset}[1]{\textit{support}(#1)}

\newcommand{\cp}{\mathbb{C}}

\newcommand{\pconc}[1]{\gamma_{\cp}(#1)}

\newcommand{\ppconcfun}{\gamma_{\ppolys}}

\newcommand{\ppconc}[1]{\ppconcfun(#1)}

\newcommand{\dom}[1]{\textit{domain}(#1)}

\newcommand{\aif}[0]{\text{\textbf{if} }}

\newcommand{\athen}[0]{\text{ \textbf{then} }}
\newcommand{\aelse}[0]{\text{ \textbf{else} }}

\newcommand{\aneg}[1]{\neg #1}

\newcommand{\sskip}{\mathsf{skip}}
\newcommand{\sifk}{\mathsf{if}}
\newcommand{\spifk}{\mathsf{pif}}
\newcommand{\sthenk}{\mathsf{then}}

\newcommand{\sif}[3]{\sifk\;{#1}\;\sthenk\;{#2}\;\mathsf{else}\;{#3}}

\newcommand{\spif}[3]{\spifk\;{#1}\;\sthenk\;{#2}\;\mathsf{else}\;{#3}}

\newcommand{\sassign}[2]{{#1}\;:=\;{#2}}
\newcommand{\sseq}[2]{{#1} \;;\; {#2}}
\newcommand{\swhile}[2]{\mathsf{while}\;{#1}\;\mathsf{do}\;{#2}}

\newcommand{\stmt}{\mathit{S}}
\newcommand{\aexp}{\mathit{E}}
\newcommand{\bexp}{\mathit{B}}
\newcommand{\relop}{\mathit{relop}}
\newcommand{\arithop}{\mathit{aop}}
\newcommand{\aop}[2]{#1\;\arithop\;#2}
\newcommand{\bop}[2]{#1 \;\relop\; #2}

\newcommand{\vars}[0]{\textbf{Var}}
\newcommand{\states}[0]{\textbf{State}}
\newcommand{\dists}[0]{\textbf{Dist}}

\newcommand{\pmass}[1]{\mathord{\parallel}#1\mathord{\parallel}}
\newcommand{\normal}[1]{\text{normal}(#1)}

\newcommand{\eeval}[2]{[\![{#1}]\!]{#2}}

\newcommand{\evalp}[2]{[\![{#1}]\!]{#2}}
\newcommand{\pevalp}[2]{[\![{#1}]\!]{#2}}

\newcommand{\paren}[1]{\left( #1 \right)}
\newcommand{\bparen}[1]{\left[ #1 \right]}

\newcommand{\ra}{\rightarrow}

\newcommand{\lsep}[0]{. \;}

\newcommand{\steps}[1][1]{\longrightarrow^{#1}}
\newcommand{\config}[2][\Pi]{\langle {#1}, {#2} \rangle}
\newcommand{\esubst}[3]{{#1}[{#2}\mapsto {#3}]}

\newcommand{\Integer}{\mathbb{Z}}
\newcommand{\Rational}{\mathbb{Q}}
\newcommand{\Real}{\mathbb{R}}

\newcommand{\set}[1]{\left\{ #1 \right\}}
\newcommand{\setsize}[1]{\left| #1 \right|}

\newcommand{\given}{\;:\;}

\newcommand{\dcond}[2]{#1 \wedge #2}

\newcommand{\drevise}[2]{#1 | #2}

\newcommand{\poly}{C}
\newcommand{\cons}{B}

\newcommand{\ppoly}{P}

\newcommand{\pp}[1]{P_{#1}}

\newcommand{\ppolys}{\mathbb{P}}
\newcommand{\ppp}[1]{\getpoly{1}}

\newcommand{\pmin}[1]{\mathrm{p}^\mathrm{min}_{#1}}
\newcommand{\pmax}[1]{\mathrm{p}^\mathrm{max}_{#1}}
\newcommand{\smin}[1]{\mathrm{s}^\mathrm{min}_{#1}}
\newcommand{\smax}[1]{\mathrm{s}^\mathrm{max}_{#1}}
\newcommand{\mmin}[1]{\mathrm{m}^\mathrm{min}_{#1}}
\newcommand{\mmax}[1]{\mathrm{m}^\mathrm{max}_{#1}}

\newcommand{\getpoly}[1]{\poly_{#1}}

\newcommand{\psize}[1]{\#(#1)}

\newcommand{\abspevalp}[2]{\langle\!\langle{#1}\rangle\!\rangle\,{#2}}

\newcommand{\ineq}{\beta}

\newcommand{\project}[2]{#1 \downharpoonright #2}
\newcommand{\forget}[2]{\mathrm{f}_{#1}(#2)}

\newcommand{\deleted}[1]{}

\newcommand{\lfp}{\mathrm{lfp}}

\ifsubmit
\newcommand{\comment}[1]{}
\newcommand{\mwh}[1]{}
\newcommand{\sbmcomment}[1]{}
\newcommand{\iscomment}[1]{}
\newcommand{\jmccomment}[1]{}
\else
\newcommand{\comment}[1]{\textcolor{blue}{#1}}
\newcommand{\mwh}[1]{\textcolor{blue}{MWH: #1}}
\newcommand{\sbmcomment}[1]{\textcolor{violet}{SBM: #1}}
\newcommand{\iscomment}[1]{\textcolor{blue}{IS: #1}}
\newcommand{\jmccomment}[1]{\textcolor{blue}{JMC: #1}}
\fi

\usepackage{xspace}

\makeatletter
\providecommand{\bigsqcap}{%
  \mathop{%
    \mathpalette\@updown\bigsqcup
  }%
}
\newcommand*{\@updown}[2]{%
  \rotatebox[origin=c]{180}{$\m@th#1#2$}%
}
\makeatother

\DeclareFontFamily{U}{MnSymbolC}{}
\DeclareSymbolFont{MnSyC}{U}{MnSymbolC}{m}{n}
\DeclareFontShape{U}{MnSymbolC}{m}{n}{
    <-6>  MnSymbolC5
   <6-7>  MnSymbolC6
   <7-8>  MnSymbolC7
   <8-9>  MnSymbolC8
   <9-10> MnSymbolC9
  <10-12> MnSymbolC10
  <12->   MnSymbolC12%
}{}
\DeclareMathSymbol{\powerset}{\mathord}{MnSyC}{180}

\begin{document}

\input glyphtounicode.tex
\pdfgentounicode=1

\title{What's the Over/Under?\\ Probabilistic Bounds on Information Leakage}

\author{
Ian Sweet$^1$ \and
Jos\'e Manuel Calder\'on Trilla$^2$ \and
Chad Scherrer$^2$ \and
Michael Hicks$^1$ \and
Stephen Magill$^2$
}
\institute{~$^1$University of Maryland and~$^2$Galois Inc.}

\maketitle

\begin{abstract}
  Quantitative information flow (QIF) is concerned with measuring
  how much of a secret is leaked to an adversary who observes the
  result of a computation that uses it. Prior work has shown that QIF
  techniques based on \emph{abstract interpretation} with
  \emph{probabilistic polyhedra} can be used to analyze the worst-case
  leakage of a query, on-line, to determine whether that query can be
  safely answered. While this approach can provide precise estimates,
  it does not scale well. This paper shows how to solve the
  scalability problem by augmenting the baseline technique with
  \emph{sampling} and \emph{symbolic execution}. We prove that our
  approach never underestimates a query's leakage (it is
  sound), and detailed experimental results show that we can
  match the precision of the baseline technique but with orders of
  magnitude better performance.
\end{abstract}



\section{Introduction}\label{sec:intro}

As more sensitive data is created, collected, and analyzed, we face the problem of
how to productively use this data while preserving privacy.
One approach to this problem is to analyze a query $f$
in order to \emph{quantify} how much information about
secret input $s$ is leaked by the output $f(s)$.
More precisely, we can consider a querier to have some \emph{prior
belief} of the secret's possible values. The belief can be modeled as a
probability distribution~\cite{clarkson09quantifying}, i.e., a
function $\delta$ from each possible value of $s$ to its probability. When a querier
observes output $o = f(s)$,
he \emph{revises} his belief, using Bayesian inference, to
produce a \textit{posterior} distribution $\delta'$.
If the posterior could reveal too much about the
secret, then the query should be rejected. One common definition
of ``too much'' is \emph{Bayes Vulnerability}, which is the
probability of the adversary guessing the secret in one
try~\cite{smith09foundations}. Formally,
$$V(\delta) \defeq \text{max}_i~\delta(i)$$
Various works~\cite{mardziel13belieflong,Besson2014,Guarnieri17,Kucera:2017:SPP:3133956.3134079}
propose rejecting $f$ if there exists an
output that makes the vulnerability of the posterior exceed a fixed
threshold $K$. In particular, for all possible values $i$ of $s$ (i.e., $\delta(i) > 0$), if
the output $o = f(i)$ could induce a posterior $\delta'$ with
$V(\delta') > K$, then the query is rejected.


One way to implement this approach is to estimate $f(\delta)$---the
distribution of $f$'s outputs when the inputs are
distributed according to $\delta$---by viewing $f$ as a
program in a \emph{probabilistic programming language}
(PPL)~\cite{Gordon:2014:PP:2593882.2593900}. Unfortunately, as
discussed in Section~\ref{sec:related}, most PPLs are approximate in a
manner that could easily result in \emph{underestimating} the
vulnerability, leading to an unsafe security decision. Techniques
designed specifically to quantify information leakage often assume
only uniform priors, cannot compute vulnerability (favoring, for
example, Shannon entropy), and/or cannot
maintain assumed knowledge between queries.

Mardziel et al.~\cite{mardziel13belieflong} propose a
\emph{sound} analysis technique based on abstract
interpretation~\cite{CousotCousot77}. In particular, they estimate a
program's probability distribution using an abstract domain called a
\emph{probabilistic polyhedron} (PP), which pairs a standard numeric
abstract domain, such as \emph{convex
  polyhedra}~\cite{Cousot:1978:ADL:512760.512770}, with some
additional \emph{ornaments}, which include lower and upper bounds on
the size of the support of the distribution, and bounds on the
probability of each possible secret value. Using PP can yield a precise, yet
safe, estimate of the vulnerability, and allows the posterior PP
(which is not necessarily uniform) to be used as a prior for the next query.
Unfortunately, PPs can be very inefficient. Defining \emph{intervals}~\cite{cousot76static} as the
PP's numeric domain can dramatically improve performance, but only with an
unacceptable loss of precision.

In this paper we present a new approach that ensures a better balance of both precision and
performance in vulnerability computation, augmenting PP with two new
techniques. In both cases we begin by analyzing a query using the fast
interval-based analysis.
Our first technique is then to use \emph{sampling} to
augment the result. In particular, we execute the query using possible
secret values $i$ sampled from the posterior $\delta'$ derived from a
particular output $o_i$. If the analysis were perfectly accurate,
executing $f(i)$ would produce $o_i$. But since intervals are
overapproximate, sometimes it will not. With many sampled outcomes,
we can construct a Beta distribution to estimate the size of
the support of the posterior, up to some level of confidence. We can
use this estimate to boost the lower bound of the abstraction, and
thus improve the precision of the estimated vulnerability.

Our second technique is of a similar flavor, but uses symbolic
reasoning to magnify the impact of a successful sample. In particular, we
execute a query result-consistent sample
\emph{concolically}~\cite{Sen:2005:CCU:1081706.1081750}, thus maintaining a
symbolic formula (called the \emph{path condition}) that characterizes
the set of variable valuations that would cause execution to follow the observed
path. We then count the number of possible
solutions and use the count to boost the
lower bound of the support (with 100\% confidence).

Sampling and concolic execution can be combined for even greater precision.

We have formalized and proved our techniques are sound
(Sections~\ref{sec:formalism}--\ref{sec:concolic}) and implemented and
evaluated them (Sections~\ref{sec:impl} and~\ref{sec:exp}). Using a
privacy-sensitive ship planning scenario
(Section~\ref{sec:overview}) we find that our techniques provide
  \emph{similar precision to convex polyhedra while providing
  orders-of-magnitude better performance}. As far as we are aware
(Section~\ref{sec:related}), our approach constitutes the best balance
of precision and performance proposed to date for estimating query
vulnerability. Our implementation freely available at
\url{https://github.com/GaloisInc/TAMBA}.

\section{Overview}
\label{sec:overview}


To provide an overview of our approach, we will describe the
application of our techniques to a scenario that involves a
coalition of ships from various nations operating in a shared region.
Suppose a natural disaster has impacted some islands in
the region.  Some number of individuals need to be evacuated from the
islands, and it falls to a regional disaster response coordinator to
determine how to accomplish this.  While the coalition wants to
collaborate to achieve these humanitarian aims, we assume that each
nation also wants to protect their sensitive data---namely ship
locations and capacity.

\begin{figure}[t]
\centering
\begin{tabular}{l@{~}|@{~}l@{~~}lc}
Field & Type & Range & Private? \\
\hline
ShipID & Integer & 1--10 & No \\
NationID & Integer & 1--20 & No \\
Capacity & Integer & 0--1000 & Yes \\
Latitude & Integer & -900,000--900,000 & Yes \\
Longitude & Integer & -1,800,000--1,800,000 & Yes \\ \hline
\end{tabular}
\caption{\label{fig:data}The data model used in the evacuation scenario.}
\end{figure}

More formally, we assume the use of the data model shown in
Figure~\ref{fig:data}, which considers a set of ships, their
coalition affiliation, the evacuation capacity of the ship, and its
position, given in terms of latitude and longitude.\footnote{We give latitude
and longitude values as integer representations of \emph{decimal
  degrees} fixed to four decimal places; e.g.,  14.3579 decimal degrees
is encoded as 143579.} We sometimes refer to the
latter two as a location $L$, with $L.x$ as the longitude and $L.y$ as
the latitude. We will often index properties by
ship ID, writing $\text{Capacity}(z)$ for the capacity associated with
ship ID $z$, or $\text{Location}(z)$ for the location.

The \textbf{evacuation problem} is defined as follows
\begin{quote}
  Given a target location $L$ and number of people to evacuate $N$,
  compute a set of nearby ships $S$ such that
  $\sum_{z \in S} \text{Capacity}(z) \geq N$.
\end{quote}
Our goal is to solve this problem in a way that minimizes the
vulnerability to the coordinator of private information, i.e., the
ship locations and their exact capacity.  We assume that this
coordinator initially has no knowledge of the positions or
capabilities of the ships other than that they fall within certain
expected ranges.

If all members of the coalition share all of their data with the
coordinator, then a solution is easy to compute, but it affords no
privacy. Figure~\ref{fig:bin-search} gives an algorithm the response
coordinator can follow that does not require each member to share all
of their data. Instead, it iteratively performs queries
\textit{AtLeast} and \textit{Nearby}. These queries do not
reveal precise values about ship locations or capacity, but rather
admit ranges of possibilities. The algorithm
works by maintaining upper and lower bounds on the capacity of each
ship \lstinline{i} in the array \lstinline{berths}.  Each ship's
bounds are updated based on the results of queries about its capacity
and location. These queries aim to be privacy preserving, doing a sort
of binary search to narrow in on the capacity of each ship in the
operating area. The procedure completes once \lstinline{is_solution}
determines the minimum required capacity is reached.

\begin{figure}[t]
\centering
\begin{minipage}{5.5in}
\begin{lstlisting}[numbers=none]
(* S = #ships; N = #evacuees; L = island loc.; D = min. proximity to L *)
      let berths = Array.make S (0,1000)
      let is_solution () = sum (Array.map fst berths) >= N
      let mid (x,y) = (x + y) / 2
      let <@\textit{AtLeast(z,b)}@> = Capacity(z) >= b
      let <@\textit{Nearby(z,l,d)}@> = |Loc(z).x - l.x| + |Loc(z).y - l.y| <= d
      while true do
        for i = 0 to S do
          let ask = mid berths[i]
          let ok = <@\textit{AtLeast}@>(i,ask) && <@\textit{Nearby}@>(i,L,D)
          if ok then  berths[i] <- (ask, snd berths[i])
          else        berths[i] <- (fst berths[i], ask)
          if is_solution () then return berths
        done
      done
\end{lstlisting}
\end{minipage}
\caption{Algorithm to solve the evacuation problem for a single island.}
\label{fig:bin-search}
\end{figure}


\subsection{Computing vulnerability with abstract interpretation}

Using this procedure, what is revealed about the private variables
(location and capacity)? Consider a single $\mathit{Nearby}(z,l,d)$ query.
At the start, the coordinator is assumed to know only that $z$ is
somewhere within the operating region. If the query returns
\lstinline{true}, the coordinator now knows that $s$ is within $d$
units of $l$ (using Manhattan distance). This makes
$\textit{Location}(z)$ more vulnerable because the adversary has less
uncertainty about it.


Mardziel et al.~\cite{mardziel13belieflong} proposed a static analysis
for analyzing queries such as $\mathit{Nearby}(z,l,d)$ to estimate the
worst-case vulnerability of private data.
If the worst-case vulnerability is too great, the query can
be rejected. A key element of their approach is to perform abstract
interpretation over the query using an abstract domain called a
\textit{probabilistic polyhedron}. An element $P$ of this domain
represents the set of possible distributions over the query's state. This
state includes both the hidden secrets and the visible query
results. The abstract interpretation is sound in the sense that the
true distribution $\delta$ is contained in the set of distributions
represented by the computed probabilistic polyhedron $P$.

A probabilistic polyhedron $P$ is a tuple comprising a \emph{shape}
and three \emph{ornaments}. The shape $C$ is an element of a standard
numeric domain---e.g., intervals~\cite{cousot76static},
octagons~\cite{mine01octagon}, or
convex polyhedra~\cite{Cousot:1978:ADL:512760.512770}---which
overapproximates the set of possible values in the support of the
distribution. The ornaments $p \in [0, 1]$, $m \in \mathbb{R}$, and $s \in \mathbb{Z}$
are pairs which store upper and lower bounds on the
probability per point, the total mass, and number of support points in
the distribution, respectively. (Distributions represented by $P$
are not necessarily normalized, so the mass $m$ is not always $1$.)

Figure~\ref{fig:prob-AI}(a) gives an example probabilistic polyhedron
that represents the posterior of a \textit{Nearby} query that returns true. In
particular, if \textit{Nearby(z,$L_1$,D)} is true then Location($z$) is
somewhere within the depicted diamond around $L_1$. Using convex
polyhedra or octagons for the shape domain would permit representing
this diamond exactly; using intervals would overapproximate it as the
depicted 9x9 bounding box. The ornaments would be the same in any
case: the size $s$ of the support is 41 possible (x,y) points,
the probability $p$ per point is $0.01$, and the total mass
is $0.41$, i.e., $p \cdot s$. In general, each ornament is a pair of
a lower and upper bound (e.g., $s_\text{min}$ and $s_\text{max}$), and $m$ might be a more accurate estimate
than $p \cdot s$. In this case shown in the figure, the bounds are
tight.

Mardziel et al's procedure works by computing the posterior $P$ for
each possible query output $o$, and from that posterior determining the
vulnerability. This is easy to do. The upper bound $p_\text{max}$
of $p$ maximizes the probability of any given point. Dividing this by
the \emph{lower bound} $m_\text{min}$ of the probability mass $m$ normalizes this
probability for the worst case. For $P$ shown in Figure~\ref{fig:prob-AI}(a), the bounds of $p$ and
$m$ are tight, so the vulnerability is simply $0.01 / 0.41 = 0.024$.


\subsection{Improving precision with sampling and concolic execution}
\label{sec:samp_and_conc}

\begin{figure}
\centering
\begin{tabular}{c}
\includegraphics[width=0.5\columnwidth]{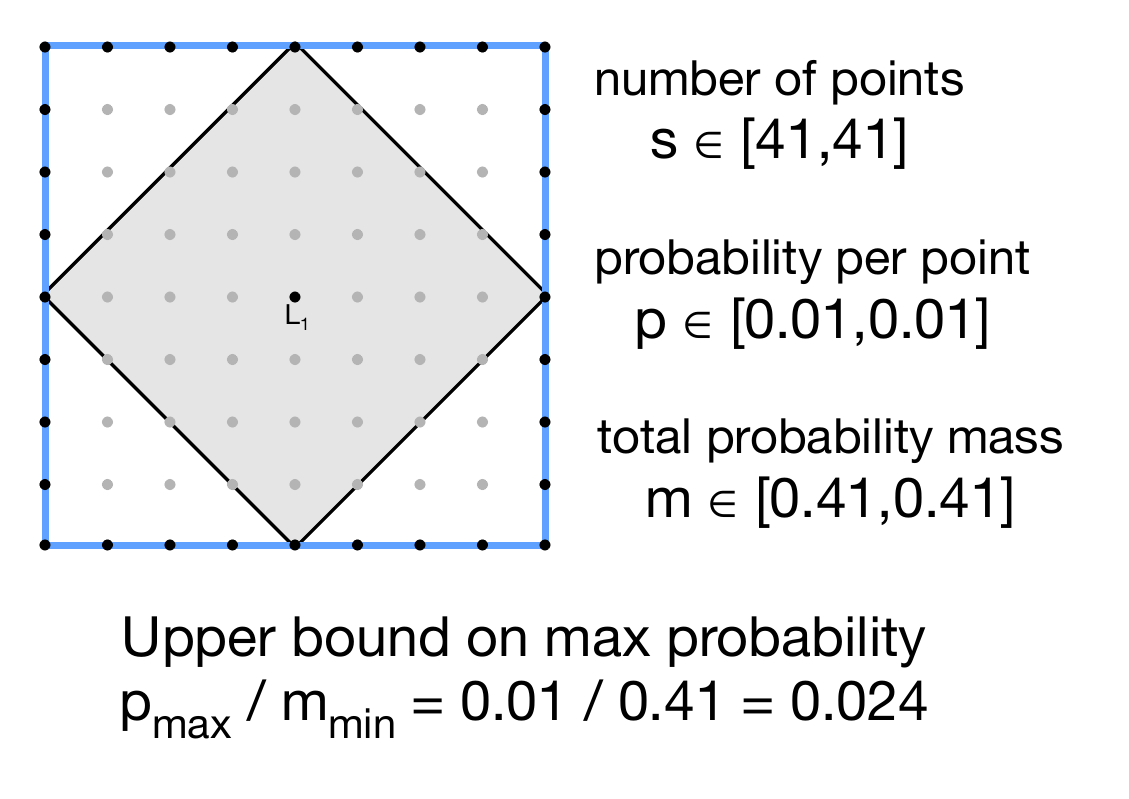}\\
\\
(a) Probabilistic polyhedra \\
\\
\includegraphics[width=\columnwidth]{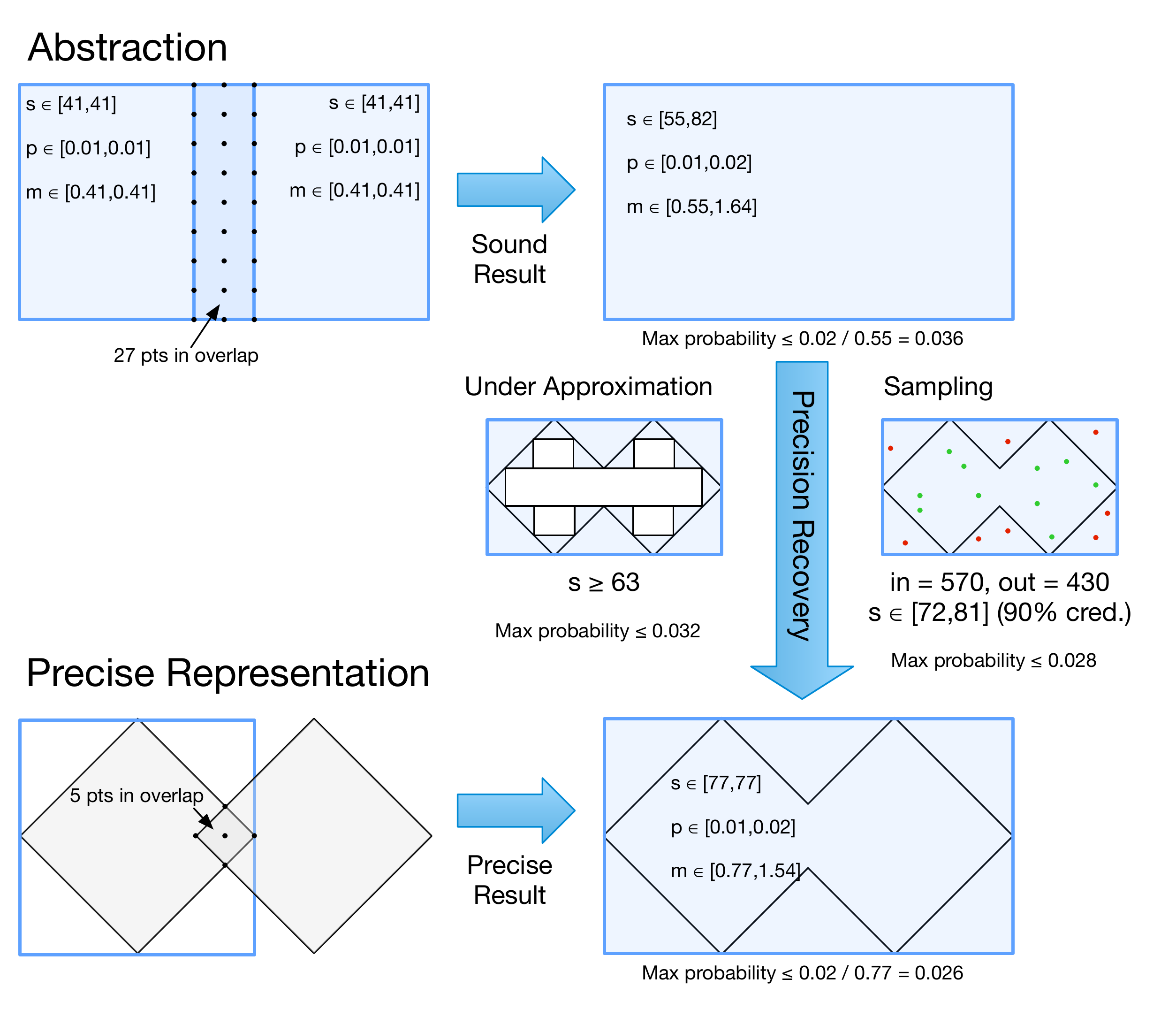}\\
\\
(b) Improving precision with sampling and underapproximation
(concolic execution)
\\
\end{tabular}
\caption{Computing vulnerability (max probability) using abstract
  interpretation}
\label{fig:prob-AI}
\end{figure}

In Figure~\ref{fig:prob-AI}(a), the parameters $s$, $p$, and $m$ are
precise.  However, as additional operations are performed, these
quantities can accumulate imprecision.  For example, suppose we are
using intervals for the shape domain, and we wish to analyze the query
$\textit{Nearby}(z,L_1,4) \vee \textit{Nearby}(z,L_2,4)$ (for some
nearby point $L_2$). The result is produced by analyzing the two
queries separately and then combining them with an \emph{abstract
  join}; this is shown in the top row of
Figure~\ref{fig:prob-AI}(b). Unfortunately, the result is very
imprecise. The bottom row of Figure~\ref{fig:prob-AI}(b) illustrates
the result we would get by using convex polyhedra as our shape
domain. When using intervals (top row), the vulnerability is estimated
as 0.036, whereas the precise answer (bottom row) is actually
0.026. Unfortunately, obtaining this precise answer is far more
expensive than obtaining the imprecise one.

This paper presents two techniques that can allow us to use the less
precise interval domain but then \emph{recover} lost precision in a
relatively cheap post-processing step. The effect of our techniques is
shown in the middle-right of Figure~\ref{fig:prob-AI}(b). Both
techniques aim to obtain better lower bounds for $s$.  This allows us
to update lower bounds on the probability mass $m$ since
$m_\text{min}$ is at least $s_\text{min} \cdot p_\text{min}$ (each
point has at least probability $p_\text{min}$ and there are at least
$s_\text{min}$ of them). A larger $m$ means a smaller vulnerability.

The first technique we explore is \emph{sampling}, depicted to the right of
the arrow in Figure~\ref{fig:prob-AI}(b).  Sampling chooses random
points and evaluates the query on them to determine whether they are in the
support of the posterior distribution for a particular query result.
By tracking the ratio of points that produce the expected output, we can
produce an estimate of $s$, whose confidence increases as we include more
samples.  This approach is depicted in the figure, where
we conclude that $s \in [72,81]$ and $m \in [0.72,1.62]$ with 90\% confidence after taking 1000
samples, improving our vulnerability estimate to $V \leq \frac{0.02}{0.72} = 0.028$.

The second technique we explore is the use of \emph{concolic
  execution} to derive a \emph{path condition}, which is a formula
over secret values that is consistent with a query result. By
performing \emph{model counting} to estimate the number of solutions
to this formula, which are an underapproximation of the true size of
the distribution, we can safely boost the lower bound of
$s$. This approach is depicted to the left of the arrow in
Figure~\ref{fig:prob-AI}(b). The depicted shapes represent discovered
path condition's disjuncts, whose size sums to 63. This is a better
lower bound on $s$ and improves the vulnerability estimate to 0.032.

These techniques can be used together to further increase
precision. In particular, we can first perform concolic execution, and
then sample from the area not covered by this
underapproximation. Importantly, Section~\ref{sec:exp} shows that
using our techniques with the interval-based analysis yields an orders of
magnitude performance improvement over using polyhedra-based analysis
alone, while achieving similar levels of precision, with high confidence.


\section{Preliminaries: Syntax and Semantics}\label{sec:formalism}

This section presents the core language---syntax and semantics---in
which we formalize our approach to computing vulnerability. We also
review \emph{probabilistic polyhedra}~\cite{mardziel13belieflong},
which is the baseline analysis technique that we augment.


\begin{figure}[t]
\[
\begin{array}{llcl}
\mathit{Variables} & x & \in & \vars \\
\mathit{Integers} & n & \in & \Integer \\
\mathit{Rationals} & q & \in & \Rational \\
\mathit{States} & \sigma & \in & \states \defeq \vars \rightharpoonup \Integer \\
\mathit{Distributions} & \delta & \in & \dists \defeq \states \rightarrow \Real+_0 \\
\mathit{Arith. ops} & \arithop &::= & + \mid \times \mid - \\
\mathit{Rel. ops} & \relop &::= & \leq \;\mid\; < \;\mid\; = \;\mid\; \neq \;\mid\; \cdots
\\
\mathit{Arith. exps} & \aexp &::= & x \mid n \mid \aop{\aexp_1}{\aexp_2} \\
\mathit{Bool. exps} & \bexp &::= & \bop{\aexp_1}{\aexp_2} \mid \bexp_1 \wedge \bexp_2 \mid \bexp_1 \vee \bexp_2 \mid \aneg{\bexp} \\

\mathit{Statements} & \stmt &::= & \sskip \mid \sassign{x}{\aexp} \mid \sseq{\stmt_1}{\stmt_2} \mid \swhile{\bexp}{\stmt} \mid \\
&     && \sif{\bexp}{\stmt_1}{\stmt_2} \mid \spif{q}{\stmt_1}{\stmt_2} \\
\end{array}
\]
\caption{Core language syntax}
\label{fig:syntax}
\end{figure}

\subsection{Core Language and Semantics}

The programming language we use for queries is given in
Figure~\ref{fig:syntax}. The language is essentially standard, apart
from $\spif{q}{\stmt_1}{\stmt_2}$, which implements probabilistic
choice: $\stmt_1$ is executed with probability $q$, and $\stmt_2$ with
probability $1-q$. We limit the form of expressions $E$ so that they
can be approximated by standard numeric abstract domains such as convex
polyhedra~\cite{Cousot:1978:ADL:512760.512770}. Such domains require
linear forms; e.g., there is no division operator and
multiplication of two variables is disallowed.\footnote{Relaxing such
limitations is possible---e.g., polynominal inequalities
can be approximated using convex
polyhedra~\cite{bagnara2005generation}---but doing so precisely and
scalably is a challenge.}

We define the semantics of a program in terms of its effect on
(discrete) distributions of states. States $\sigma$ are partial maps from
variables to integers; we write $ \dom{\sigma} $ for the set of
variables over which $\sigma$ is defined.  Distributions $\delta$
are maps from states to nonnegative real numbers, interpreted as
probabilities (in range $[0,1]$). The denotational semantics
considers a program as a relation between
distributions. In particular, the semantics of statement $\stmt$, written
$\pevalp{\stmt}{}$, is a function of the form $\dists \rightarrow \dists$;
we write $\pevalp{\stmt}{\delta} = \delta'$ to say that the semantics
of $\stmt$ maps input distribution $\delta$ to output distribution
$\delta'$. Distributions are not necessarily normalized; we write
$\pmass{\delta}$ as the probability mass of $\delta$ (which
is between 0 and 1). We write $\dot{\sigma}$ to denote the point
distribution that gives $\sigma$ probability 1, and all other
states 0.

The semantics is standard and not crucial in order to understand our
techniques.
In Appendix~\ref{sec:dist-semantics}
we provide the semantics in full, see Clarkson et
al.~\cite{clarkson09quantifying} or Mardziel et
al~\cite{mardziel13belieflong} for detailed explanations.


\subsection{Probabilistic polyhedra}

To compute vulnerability for a program $\stmt$ we must compute (an
approximation of) its output distribution. One way to do that would be
to use sampling: Choose states $\sigma$ at random from the input
distribution $\delta$, ``run'' the program using that input state, and
collect the frequencies of output states $\sigma'$ into a distribution
$\delta'$. 
While using sampling in this manner is
simple and appealing, it could be both expensive and imprecise. In
particular, depending on the size of the input and output space, it
may take many samples to arrive at a proper approximation of the
output distribution.

\emph{Probabilistic polyhedra}~\cite{mardziel13belieflong} can address
both problems. This abstract domain combines a standard
domain $\getpoly{}$ for representing numeric program states
with additional \emph{ornaments} that all together can safely
represent $\stmt$'s output distribution.

Probabilistic polyhedra
work for any numeric domain; in this paper we use both convex
polyhedra~\cite{Cousot:1978:ADL:512760.512770} and
intervals~\cite{cousot76static}. For concreteness, we present the
defintion using convex polyhedra. We use the meta-variables
$\ineq, \ineq_1, \ineq_2$, etc. to denote linear inequalities.  

\begin{definition}
  A \emph{convex polyhedron} $\poly = (\cons, V)$ is a set of linear inequalities
  $\cons = \{\ineq_1,\ldots,\ineq_m\}$, interpreted conjunctively,
  over variables $ V $.  We write
  $\cp$ for the set of all convex polyhedra.   A polyhedron $\poly$
  represents a set of states, denoted $\pconc{\poly}$, as follows, where $\sigma \models \ineq$ indicates that the state $\sigma$ satisfies the inequality $\ineq$.
\[\pconc{\paren{\cons,V}} \defeq \{\sigma \given \dom{\sigma} =
V,\; \forall \ineq \in \cons.\ \sigma \models \ineq\}\]

Naturally we require that $\dom{\{\ineq_1,\ldots,\ineq_n\}} \subseteq V
$; i.e., $V$ mentions all variables in the inequalities. Let
$ \dom{\paren{\cons, V}} = V $.

\end{definition}

Probabilistic polyhedra extend this standard representation of sets of
program states to sets of \emph{distributions} over program states.

\begin{definition} \label{def:ppoly}
A \emph{probabilistic polyhedron} $\pp{}$ is a tuple
$(\getpoly{},\smin{},$ $\smax{}, \pmin{},$ $\pmax{}, \mmin{},$ $\mmax{})$.
We write $\ppolys$ for the set of probabilistic polyhedra.  The
quantities $\smin{}$ and $ \smax{}$ are lower and upper bounds on
the number of support points in the concrete distribution(s) $\pp{}$
represents. A support point of a distribution is one
which has non-zero probability.  The quantities
$\pmin{} $ and $ \pmax{} $ are lower and upper bounds on the
probability mass per support point.  The $ \mmin{} $ and $
\mmax{} $ components give bounds on the total probability mass (i.e.,
the sum of the probabilities of all support points).
Thus $\pp{}$ represents the \emph{set} of distributions
$\ppconc{\pp{}}$ defined below.
\[
\ppconc{\pp{}} \defeq
\begin{aligned}[t]
\{\delta \given {} & \nzset{\delta} \subseteq \pconc{\getpoly{}} \wedge {} \\
  & \smin{} \leq \setsize{\nzset{\delta}} \leq \smax{} \wedge {} \\
  & \mmin{} \leq \pmass{\delta} \leq \mmax{} \wedge \\
  & \forall \sigma \in \nzset{\delta}.\ \pmin{} \leq \delta(\sigma) \leq \pmax{}\}
\end{aligned}
\]

We will write $ \dom{\pp{}} \defeq \dom{\poly} $ to denote the set of variables
used in the probabilistic polyhedron.

\end{definition}

Note the set $\ppconc{\pp{}}$ is a singleton exactly when $\smin{} =
\smax{} = \psize{\getpoly{}}$ and $\pmin{} = \pmax{}$, and $\mmin{} =
\mmax{}$, where $\psize{\getpoly{}}$ denotes the number
  of discrete points in convex polyhedron $\getpoly{}$.   In such a
case $\ppconc{\pp{}}$ contains only the uniform
distribution where each state in $\pconc{\getpoly{}}$ has probability
$\pmin{}$. In general, however, the concretization of a probabilistic
polyhedron will have an infinite number of distributions, with
per-point probabilities varied somewhere in the range $ \pmin{}
$ and $ \pmax{} $. 
%
Distributions represented by a probabilistic polyhedron
are not necessarily normalized.  In general, there is a
relationship between $\pmin{}, \smin{}, $ and $\mmin{}$, in that
$\mmin{} \geq \pmin{} \cdot \smin{}$ (and $\mmax{} \leq \pmax{} \cdot
\smax{}$), and the combination of the three can yield more information
than any two in isolation.
The \emph{abstract semantics} of $\stmt$ is written
$\abspevalp{\stmt}{\ppoly} = \ppoly'$, and indicates that abstractly
interpreting $\stmt$ where the distribution of input states are
approximated by $\ppoly$ will produce $\ppoly'$, which approximates
the distribution of output states. Following standard abstract
interpretation terminology, $\powerset{\dists}$ (sets
of distributions) is the \textit{concrete domain}, $\ppolys$ is the
\textit{abstract domain}, and
$\ppconcfun : \ppolys \rightarrow \powerset{\dists}$ is the
\textit{concretization function} for $\ppolys$. We do not present the
abstract semantics here; details can be found in Mardziel et
al.~\cite{mardziel13belieflong}. Importantly, this abstract semantics
is sound:
\begin{theorem}[Soundness]
  \label{soundness}
  For all $\stmt, \ppoly_1, \ppoly_2, \delta_1, \delta_2$, if
  $\delta_1 \in \ppconc{\ppoly_1}$ and
  $\abspevalp{\stmt}{\ppoly_1} = \ppoly_2$, then
  $\pevalp{\stmt}{\delta_1} = \delta_2$ with
  $\delta_2 \in \ppconc{\ppoly_2}$.
\end{theorem}
\begin{proof}
  See Theorem 6 in Mardziel et. al~\cite{mardziel13belieflong}.
\end{proof}

Consider the example from Section~\ref{sec:samp_and_conc}. We assume
the adversary has no prior information about the location of ship
$s$. So, $\delta_1$ above is simply the uniform distribution over all
possible locations. The statement $\stmt$ is the query issued by the
adversary, $\textit{Nearby}(z,L_1,4) \vee
\textit{Nearby}(z,L_2,4)$.\footnote{Appendix~\ref{app:code} shows the
  code, which computes Manhattan distance between $s$ and $L_1$ and $L_2$ and then
  sets an output variable if either distance is within four units.}
If we assume that the result of the query
is \lstinline|true| then the adversary learns that the location of $s$ is within (Manhattan)
distance $4$ of $L_1$ or $L_2$. This posterior belief ($\delta_2$) is represented by
the overlapping diamonds on the bottom-right of Figure~\ref{fig:prob-AI}(b).
The abstract interpretation produces a sound (interval) overapproximation ($\ppoly_2$) of the
posterior belief. This is modeled by the rectangle which surrounds the overlapping diamonds.
This rectangle is the ``join'' of two overlapping boxes, which each correspond to one of the $\textit{Nearby}$
calls in the disjuncts of $\stmt$.

\section{Computing Vulnerability: Basic procedure}
\label{sec:basic}

The key goal of this paper is to quantify the risk to secret
information of running a query over that information. This section explains the basic
approach by which we can use probabilistic polyhedra to compute
\emph{vulnerability}, i.e., the probability of the most probable point
of the posterior distribution. Improvements on this basic approach are
given in the next two sections.

Our convention will be to use $\getpoly{1}$, $\smin{1}$, $\smax{1}$,
etc. for the components associated with probabilistic polyhedron
$\pp{1}$. In the program $\stmt$ of interest, we assume that secret
variables are in the set $T$, so input states are written $\sigma_T$,
and we assume there is a single output variable $r$.  We assume that
the adversary's initial uncertainty about the possible values of the secrets
$T$ is captured by the probabilistic polyhedron $\ppoly_0$ (such that
$\dom{\ppoly_0} \supseteq T$).

Computing vulnerability occurs according to the following procedure.
\begin{enumerate}
\item Perform abstract interpretation: $\abspevalp{\stmt}{\ppoly_0} = \ppoly$
\item Given a concrete output value of interest, $o$, perform abstract
  conditioning to define $\ppoly_{r=o} \defeq (\ppoly \wedge
  r\!=\!o)$.\footnote{We write $\ppoly \wedge B$ and not $\ppoly
    \mid B$
    because $\ppoly$ need not be normalized.}
\end{enumerate}
The vulnerability $V$ is the probability of the most likely state(s).
When a probabilistic polyhedron represents one or more true
distributions (i.e., the probabilities all sum to 1), the most
probable state's probability is bounded by $\pmax{}$. However, the
abstract semantics does not always normalize the probabilistic
polyhedron as it computes, so we need to scale $\pmax{}$ according to
the total probability mass. To ensure that our estimate is on the safe
side, we scale $\pmax{}$ using the \emph{minimum} probability mass:
$V = \frac{\pmax{}}{\mmin{}}$. In Figure~\ref{fig:prob-AI}(b), the sound
approximation in the top-right has $V \leq \frac{0.02}{0.55} = 0.036$ and the most
precise approximation in the bottom-right has $V \leq \frac{0.02}{0.77} = 0.026$.


\section{Improving precision with sampling}
\label{sec:sampling}

We can improve the precision of the basic procedure using
sampling. First we introduce some notational convenience:

\begin{align*}
\ppoly_{T} &\defeq \project{\dcond{\ppoly}{(r = o)}}{T} \\
\ppoly_{T+} &\defeq \ppoly_{T} \text{ revised polyhedron with confidence }\omega
\end{align*}

$\ppoly_{T}$ is equivalent to step 2, above, but projected onto the set
of secret variables $T$. $\ppoly_{T+}$ is the improved (via sampling)
polyhedron.

After computing $\ppoly_T$ with the basic procedure from the previous section we
take the following additional steps:
\begin{enumerate}
\item Set counters $\alpha$ and $\beta$ to zero.
\item Do the following $N$ times (for some $N$, see below):
\begin{enumerate}
\item Randomly select an input state
  $\sigma_T \in \pconc{\getpoly{T}}$.
\item \label{step:sample} ``Run'' the program by computing
  $\pevalp{\stmt}{\dot{\sigma_T}} = \delta$.
  If there exists $\sigma \in \nzset{\delta}$ with $\sigma(r) = o$
  then increment $\alpha$, else increment $\beta$.
\end{enumerate}
\item We can interpret $\alpha$ and $\beta$ as the parameters of a
  Beta distribution of the likelihood that an arbitrary state in
  $\pconc{\getpoly{T}}$ is in the support of the true distribution.
  From these parameters we can
  compute the \emph{credible interval} $[p_L,p_U]$ within which is
  contained the true likelihood, with
  confidence $\omega$ (where $0 \leq \omega \leq 1$). (A credible
  interval is essentially a Bayesian analogue of a confidence
  interval.)
 In general, obtaining a higher confidence or a narrower interval will require a higher $N$.
  Let result
  $\ppoly_{T+} = \ppoly_T$ except that
  $\smin{T+} = p_L \cdot \psize{\getpoly{T}}$ and
  $\smax{T+} = p_U \cdot \psize{\getpoly{T}}$ (assuming these improve
  on $\smin{T}$ and $\smax{T}$).
  We can then propagate these improvements to $\mmin{}$ and $\mmax{}$
  by defining $\mmin{T+} = \pmin{T} \cdot \smin{T+}$ and
  $\mmax{T+} = \pmax{T} \cdot \smax{T+}$. Note that if $\mmin{T} >
  \mmin{T+}$ we leave it unchanged, and do likewise if $\mmax{T} <
  \mmax{T+}$.
\end{enumerate}
At this point we can compute the vulnerability as in the basic
procedure, but using $\ppoly_{T+}$ instead of $\ppoly_T$.

Consider the example of Section~\ref{sec:samp_and_conc}. In Figure~\ref{fig:prob-AI}(b),
we draw samples from the rectangle in the top-right. This rectangle overapproximates
the set of locations where $s$ might be, given that the query returned \lstinline{true}. We
sample locations from this rectangle and run the query on each sample. The green (red) dots
indicate \lstinline{true} (\lstinline{false}) results,
which are added to $\alpha$ ($\beta$). After
sampling $N = 1000$ locations, we have $\alpha = 570$ and $\beta = 430$.
Choosing $\omega = .9$ (90\%), we compute the credible interval $[0.53, 0.60]$. With $\psize{\getpoly{T}} = 135$,
we compute $[\smin{T+},\smax{T+}]$ as $[0.53 \cdot 135, 0.60 \cdot 135] = [72,81]$.

There are several things to notice about this procedure. First,
observe that in step~\ref{step:sample} we ``run'' the program using
the point distribution $\dot{\sigma}$ as an input; in the case that
$\stmt$ is deterministic (has no $\spifk$ statements) the output
distribution will also be a point distribution. However, for programs
with $\spifk$ statements there are multiple possible outputs depending
on which branch is taken by a $\spifk$. We consider all of these
outputs so that we can confidently determine whether the input state
$\sigma$ could ever cause $\stmt$ to produce result $o$. If so, then
$\sigma$ should be considered part of $\ppoly_{T+}$. If not, then we
can safely rule it out (i.e., it is part of the overapproximation).

Second, we only update the size parameters of $\ppoly_{T+}$; we make no changes
to $\pmin{T+}$ and $\pmax{T+}$. This is because our sampling procedure only
determines whether it is \emph{possible} for an input state to produce the
expected output. The probability that an input state produces an output state
is already captured (soundly) by $p_T$ so we do not change that.
This is useful because the approximation of $p_T$ does not degrade with the use
of the interval domain in the way the approximation of the size degrades (as
illustrated in Figure~\ref{fig:prob-AI}(b)). Using sampling is an attempt to
regain the precision lost on the size component (only). 

Finally, the confidence we have that sampling has accurately assessed
which input states are in the support is orthogonal to the probability of
any given state. In particular, $\ppoly_T$ is an abstraction of a
distribution $\delta_T$, which is a mathematical object. Confidence
$\omega$ is a measure of how likely it is that our abstraction (or, at
least, the size part of it) is accurate.

We prove (in Appendix~\ref{sec:proofs}) that our sampling procedure is sound:
\begin{theorem}[Sampling is Sound] \hfill \\
  \label{sampling_proof}
  If $\delta_0 \in \ppconc{\ppoly_0}$, $\abspevalp{\stmt}{\ppoly_0} =
  \ppoly$, and $\pevalp{\stmt}{\delta_0} = \delta$ then
  $\delta_{T} \in \ppconc{\ppoly_{T+}} \text{ with confidence } \omega$
  where
  \begin{align*}
  \delta_{T} &\defeq \project{\dcond{\delta}{(r = o)}}{T} \\
  \ppoly_{T} &\defeq \project{\dcond{\ppoly}{(r = o)}}{T} \\
  \ppoly_{T+} &\defeq \ppoly_{T} \text{ sampling revised with confidence }\omega.
  \end{align*}
  %
\end{theorem}

\section{Improving precision with concolic execution}
\label{sec:concolic}

Another approach to improving the precision of a probabilistic
polyhedron $\ppoly$ is to use concolic execution. The idea here is to
``magnify'' the impact of a single sample to soundly increase
$\smin{}$ by considering its execution \emph{symbolically}. More
precisely, we concretely execute a program using a particular secret
value, but maintain symbolic constraints about how that value is
used. This is referred to as \emph{concolic} execution~\cite{Sen:2005:CCU:1081706.1081750}. We use the
collected constraints to identify all points that would induce the
same execution path, which we can include as part of $\smin{}$.

We begin by defining the semantics of concolic execution, and then
show how it can be used to increase $\smin{}$ soundly.

\subsection{(Probabilistic) Concolic Execution}

Concolic execution is expressed as rewrite rules defining a judgment
$\config{\stmt} \steps[p]_\pi \config[\Pi']{\stmt'}$.  Here, $\Pi$ is
pair consisting of a concrete state $\sigma$ and symbolic state
$\zeta$. The latter maps variables $x \in \vars$ to \emph{symbolic
  expressions} $\mathcal{E}$ which extend expressions $E$ with
\emph{symbolic variables} $\alpha$. This judgment indicates that under
input state $\Pi$ the statement $\stmt$ reduces to statement $\stmt'$
and output state $\Pi'$ with probability $p$, with \emph{path
  condition} $\pi$.  The path condition is a conjunction of boolean
symbolic expressions $\mathcal{B}$ (which are just boolean expressions
$\bexp$ but altered to use symbolic expressions $\mathcal{E}$ instead
of expressions $E$) that record which branch is taken during
execution. For brevity, we omit $\pi$ in a rule when it is
$\mathsf{true}$.

The rules for the concolic semantics are given in
Figure~\ref{fig:symbolic-semantics}. Most of these are standard, and
deterministic (the probability annotation $p$ is $1$). Path conditions
are recorded for $\sifk$ and $\mathsf{while}$, depending on the branch
taken. The semantics of $\spif{q}{\stmt_1}{\stmt_2}$ is
non-deterministic: the result is that of $\stmt_1$ with probability
$q$, and $\stmt_2$ with probability $1 - q$. We write $\zeta(B)$ to
substitute free variables $x \in B$ with their mapped-to values
$\zeta(x)$ and then simplify the result as much as possible. For
example, if $\zeta(x) = \alpha$ and $\zeta(y) = 2$, then
$\zeta(x > y + 3) = \alpha > 5$. The same goes for $\zeta(E)$.

We define a \emph{complete run} of the concolic semantics with the
judgment $\config{\stmt} \Downarrow^p_{\pi} \Pi'$, which has two rules:
\[
\begin{array}{cc}
\config{\sskip} \Downarrow^1_{\mathsf{true}} \Pi \\ \\
\config{\stmt} \steps[p]_{\pi} \config[\Pi']{\stmt'} \quad
\config[\Pi']{\stmt'} \Downarrow^q_{\pi'} \Pi'' \\ \hline
\config{\stmt} \Downarrow^{p\cdot q}_{\pi \wedge \pi'} \Pi'' \\
\end{array}
\]
A complete run's probability is thus the product of the probability of
each individual step taken. The run's path condition is the
conjunction of the conditions of each step.

The path condition $\pi$ for a complete run is a conjunction of the
(symbolic) boolean guards evaluated during an execution. $\pi$ can be
converted to disjunctive normal form (DNF), and given the restrictions
of the language the result is essentially a set of convex polyhedra
over symbolic variables $\alpha$.


\begin{figure}[t]
\[
\begin{array}{lr}
\multicolumn{2}{l}{\config[(\sigma,\zeta)]{\sassign{x}{\aexp}} \steps %
\config[(\esubst{\sigma}{x}{\sigma(\aexp)},\esubst{\zeta}{x}{\zeta(\aexp)})]{\sskip}} \\
\config[(\sigma,\zeta)]{\sif{\bexp}{\stmt_1}{\stmt_2}} \steps_{\zeta(B)} \config[(\sigma,\zeta)]{\stmt_1} & \text{if
  }\sigma(\bexp) \\
\config[(\sigma,\zeta)]{\sif{\bexp}{\stmt_1}{\stmt_2}} \steps_{\zeta(\aneg{\bexp})} \config[(\sigma,\zeta)]{\stmt_2} & \text{if
  }\sigma(\aneg{\bexp}) \\
\config{\spif{q}{\stmt_1}{\stmt_2}} \steps[q] \config{\stmt_1}\\
\config{\spif{q}{\stmt_1}{\stmt_2}} \steps[1\!-\!q] \config{\stmt_2}\\
\multicolumn{2}{l}{\config{\sseq{\stmt_1}{\stmt_2}} \steps_\pi
  \config[\Pi']{\sseq{\stmt_1'}{\stmt_2}} \quad
\text{if }\config{\stmt_1} \steps_\pi \config[\Pi']{\stmt_1'}} \\
\config{\sseq{\sskip}{\stmt}} \steps \config{\stmt}\\
\config{\swhile{\bexp}{\stmt}} \steps_{\zeta(\bexp)} \config{\sseq{\stmt}{\swhile{\bexp}{\stmt}}} & \text{if }\sigma(\bexp)\\
\config{\swhile{\bexp}{\stmt}} \steps_{\zeta(\aneg{\bexp})} \config{\sskip} & \text{if }\sigma(\aneg{\bexp})\\
\end{array}
\]
\caption{Concolic semantics}
\label{fig:symbolic-semantics}
\end{figure}



\subsection{Improving precision}

Using concolic execution, we can improve our estimate of the size of a
probabilistic polyhedron as follows:
\begin{enumerate}
\item  Randomly select an input state  $\sigma_T \in
  \pconc{\getpoly{T}}$ (recall that $\getpoly{T}$ is the polyhedron describing
  the possible valuations of secrets $T$).
\item Set $\Pi = (\sigma_T,\zeta_T)$ where $\zeta_T$ maps each
  variable $x \in T$ to a fresh symbolic variable $\alpha_x$. Perform
  a complete concolic run
  $\config[\Pi]{\stmt} \Downarrow^p_\pi (\sigma',\zeta')$. Make sure
  that $\sigma'(r) = o$, i.e., the expected output. If not, select a new
  $\sigma_T$ and retry. Give up after some number of failures $N$.
  For our example shown in Figure~\ref{fig:prob-AI}(b), we might obtain a
  path condition $|Loc(z).x - L_1.x| + |Loc(z).y - L_1.y| \leq 4$ that
  captures the left diamond of the disjunctive query.
\item After a successful concolic run, convert path condition $\pi$ to
  DNF, where each conjunctive
  clause is a polyhedron $\getpoly{i}$. Also convert uses of
  disequality ($\leq$ and $\geq$) to be strict ($<$ and
  $>$).
\item Let $\getpoly{} = \getpoly{T} \sqcap (\bigsqcup_i \getpoly{i}$);
  that is, it is the join of each of the polyhedra in $DNF(\pi)$
  ``intersected'' with the original constraints. This captures all of the
  points that could possibly lead to the observed outcome along the
  concolically executed path.  Compute $n = \psize{\getpoly{}}$.  Let
  $\ppoly_{T+} = \ppoly_T$ except define $\smin{T+} = n$ if $\smin{T} <
  n$ and $\mmin{T+} = \pmin{T} \cdot n$ if $\mmin{T} < \pmin{T} \cdot
  n$. (Leave them as is, otherwise.) For our example, $n = 41$, the
  size of the left diamond. We do not update
  $\smin{T}$ since 41 < 55, the probabilistic polyhedron's lower
  bound (but see below).
\end{enumerate}

\begin{theorem}[Concolic Execution is Sound] \hfill \\
  If $\delta_0 \in \ppconc{\ppoly_0}$, $\abspevalp{\stmt}{\ppoly_0} = \ppoly$, and $\pevalp{\stmt}{\delta_0} = \delta$ then
$\delta_{T} \in \ppconc{\ppoly_{T+}} $
  where
  \begin{align*}
  \delta_{T} &\defeq \project{\dcond{\delta}{(r = o)}}{T} \\
  \ppoly_{T} &\defeq \project{\dcond{\ppoly}{(r = o)}}{T} \\
  \ppoly_{T+} &\defeq \ppoly_{T} \text{ concolically revised.}
  \end{align*}
\end{theorem}
The proof is in Appendix~\ref{sec:proofs}.

\subsection{Combining Sampling with Concolic Execution}
\label{sec:comb}

Sampling can be used to further augment the results of concolic execution.
The key insight is that the presence of a sound under-approximation
generated by the concolic execution means that it is unnecessary to
sample from the under-approximating region. Here is the algorithm:

\begin{enumerate}
\item Let $\getpoly{} = \getpoly{0} \sqcap (\bigsqcup_i \getpoly{i})$ be the
  under-approximating region.
\item Perform sampling per the algorithm in
  Section~\ref{sec:sampling}, but with two changes:
\begin{itemize}
\item if a sampled state $\sigma_T \in \pconc{\getpoly{}}$, ignore it
\item When done sampling, compute $\smin{T+} = p_L \cdot
  (\psize{\getpoly{T}} - \psize{\getpoly{}}) + \psize{\getpoly{}}$ and
  $\smax{T+} = p_U \cdot (\psize{\getpoly{T}} - \psize{\getpoly{}}) +
  \psize{\getpoly{}}$. This differs from Section~\ref{sec:sampling} in not
  including the count from concolic region $\getpoly{}$ in the
  computation. This is because, since we ignored samples $\sigma_T \in
  \pconc{\getpoly{}}$, the credible interval $[p_L,p_U]$ bounds
  the likelihood that any given point in $\getpoly{T} \setminus \getpoly{}$ is in the support of
  the true distribution.
\end{itemize}
\end{enumerate}
For our example, concolic execution indicated there are at least
41 points that satisfy the query. With this in hand, and using the same
samples as shown in Section~\ref{sec:sampling}, we can refine $s \in [74,80]$ and
$m \in [0.74,0.160]$ (the credible interval is formed over only those samples which satisfy
the query but fall outside the under-approximation returned by concolic execution). We improve
the vulnerability estimate to $V \leq \frac{0.02}{0.0.74} = 0.027$. These bounds (and vulnerability
estimate) are better than those of sampling alone ($s \in [72,81]$ with $V \leq 0.028$).

The statement of soundness and its proof
can be found in Appendix~\ref{sec:dist-semantics}.

\section{Implementation}\label{sec:impl}


We have implemented our approach as an extension of Mardziel et
al.~\cite{mardziel13belieflong}, which is written in OCaml. This
baseline implements numeric domains $\getpoly{}$ via
an OCaml interface to the Parma Polyhedra
Library~\cite{Bagnara:2008:PPL:1385689.1385711}. The counting procedure
$\psize{\getpoly{}}$ is implemented by LattE~\cite{latte}. Support for
arbitrary precision and exact arithmetic (e.g., for manipulating
$\mmin{}$, $\pmin{}$, etc.) is provided by the \texttt{mlgmp} OCaml
interface to the GNU Multi Precision Arithmetic library. Rather than maintaining a single probabilistic
polyhedron $\ppoly$, the implementation maintains a \emph{powerset} of
polyhedra~\cite{bagnara06powerset}, i.e., a finite disjunction. Doing
so results in a more precise handling of join points in the control flow,
at a somewhat higher performance cost.

We have implemented our extensions to this baseline for the case that
domain $\getpoly{}$ is the interval numeric domain~\cite{cousot76static}. Of
course, the theory fully applies to any numeric abstract domain. We
use Gibbs sampling, which we implemented ourselves. We delegate the
calculation of the beta distribution and its corresponding credible
interval to the \texttt{cephes} OCaml library, which in turns uses the
GNU Scientific Library. It is straightforward to lift the various
operations we have described to the powerset domain. All of our code
is available at \url{https://github.com/GaloisInc/TAMBA}.

\section{Experiments}
\label{sec:exp}

To evaluate the benefits of our techniques, we applied them to queries
based on the evacuation problem outlined in
Section~\ref{sec:overview}. We found that while the baseline technique
can yield precise answers when computing vulnerability, our new
techniques can achieve close to the same level of precision far more
efficiently.

\subsection{Experimental Setup}

For our experiments we analyzed queries similar to
$\mathit{Nearby}(s,l,d)$ from Figure~\ref{fig:bin-search}. We generalize
the \textit{Nearby} query to accept a set of locations $L$---the
query returns true if $s$ is within $d$ units of any one of the
islands having location $l \in L$. In our experiments we fix
$d = 100$. We consider the secrecy of the location of $s$, $\mathit{Location}(s)$.
We also analyze the execution of the resource allocation
algorithm of Figure~\ref{fig:bin-search} directly; we discuss this in
Section~\ref{sec:evac}.

We measure the time it takes to compute the \emph{vulnerability} (i.e.,
the probability of the most probable point) following each query. In our experiments, we
consider a single ship $s$ and set its coordinates so that it is always in
range of some island in $L$, so that the concrete query result returns
\lstinline{true} (i.e. $\mathit{Nearby}(s, L, 100) = true$). We measure the vulnerability following this query
result starting from a prior belief that the coordinates of $s$ are
uniformly distributed with
$0 \leq \mathrm{Location}(s).x \leq 1000$ and
$0 \leq \mathrm{Location}(s).y \leq 1000$.

In our experiments, we varied several experimental parameters:
\emph{analysis method} (either P, I, CE, S, or CE+S), \emph{query
  complexity} $c$; \emph{AI precision
  level} $p$; and \emph{number of samples} $n$. We describe each in
turn.

\paragraph*{Analysis method}
We compared five techniques for computing vulnerability:
\vspace*{-.1in}
\begin{description}
\item[\textbf{P}:] Abstract interpretation (AI) with convex
  polyhedra for domain $\getpoly{}$ (Section~\ref{sec:basic}),
\item[\textbf{I}:] AI with intervals for $\getpoly{}$ (Section~\ref{sec:basic}),
\item[\textbf{S}:] AI with intervals augmented with sampling (Section~\ref{sec:sampling}),
\item[\textbf{CE}:] AI with intervals augmented with concolic execution
  (Section~\ref{sec:concolic}), and
\item[\textbf{CE+S}:] AI with intervals augmented with both techniques
  (Section~\ref{sec:comb})
\end{description}
The first two techniques are due to Mardziel et
al.~\cite{mardziel13belieflong}, where the former uses convex
polyhedra and the latter uses intervals (aka boxes) for the underlying
polygons. In our experiments we tend to focus on P since I's precision
is unacceptably poor (e.g., often vulnerability = 1).



\textit{Query complexity.}
We consider queries with different $L$; we say we are increasing the
\emph{complexity} of the query as $L$ gets larger. Let $c = |L|$; we
consider $1 \leq c \leq 5$, where larger $L$ include the same
locations as smaller ones. We set each location to be at
least $2 \cdot d$ Manhattan distance units away from any other island
(so diamonds like those in Figure~\ref{fig:prob-AI}(a) never
overlap).

\textit{Precision.}
The precision parameter $p$ bounds the size of the powerset abstract
domain at all points during abstract interpretation. This has the
effect of forcing joins when the powerset grows larger than the
specified precision. As $p$ grows larger, the results of
abstract interpretation are likely to become more precise
(i.e. vulnerability gets closer to the true value). We considered
$p$ values of $1$, $2$, $4$, $8$, $16$, $32$, and $64$.

\textit{Samples taken.}
For the latter three analysis methods, we varied the number of samples
taken $n$. For analysis CE, $n$ is
interpreted as the number of samples to try per polyhedron before
giving up trying to find a ``valid sample.''\footnote{This is the $N$
  parameter from section~\ref{sec:concolic}.} For analysis S,
$n$ is the number of samples, distributed
proportionally across all the polyhedra in the powerset.  For analysis
CE+S, $n$ is the combination of the two. We considered sample size
values of $1,000 - 50,000$ in increments of $1,000$. We always compute
an interval with $\omega=$99.9\% confidence  (which will be wider when fewer samples
are used).

\textit{System description.}
 We ran experiments varying all
possible parameters. For each run, we measured the total execution
time (wall clock) in seconds to analyze the query and compute vulnerability. All
experiments were carried out on a MacBook Air with OSX version
10.11.6, a 1.7GHz Intel Core i7, and 8GB of RAM. We ran a single trial for each
configuration of parameters. Only wall-clock time varies across
trials; informally, we observed time variations to be small.

\subsection{Results}

Figure~\ref{fig:results}(a)--(c) measure
vulnerability (y-axis) as a function of time (x-axis)
for each analysis.\footnote{These are best viewed on a color display.}
These three figures characterize three interesting ``zones'' in the
space of complexity and precision. The results for method I are not shown in any of the figures.
This is because I always produces a vulnerability of $1$. The refinement methods (CE, S, and CE+S)
are all over the interval domain, and should be considered as ``improving'' the vulnerability of I.

\begin{figure}[t]
\begin{tabular}{cc}
  \begin{minipage}{.5\columnwidth}
  \includegraphics[width=\columnwidth]{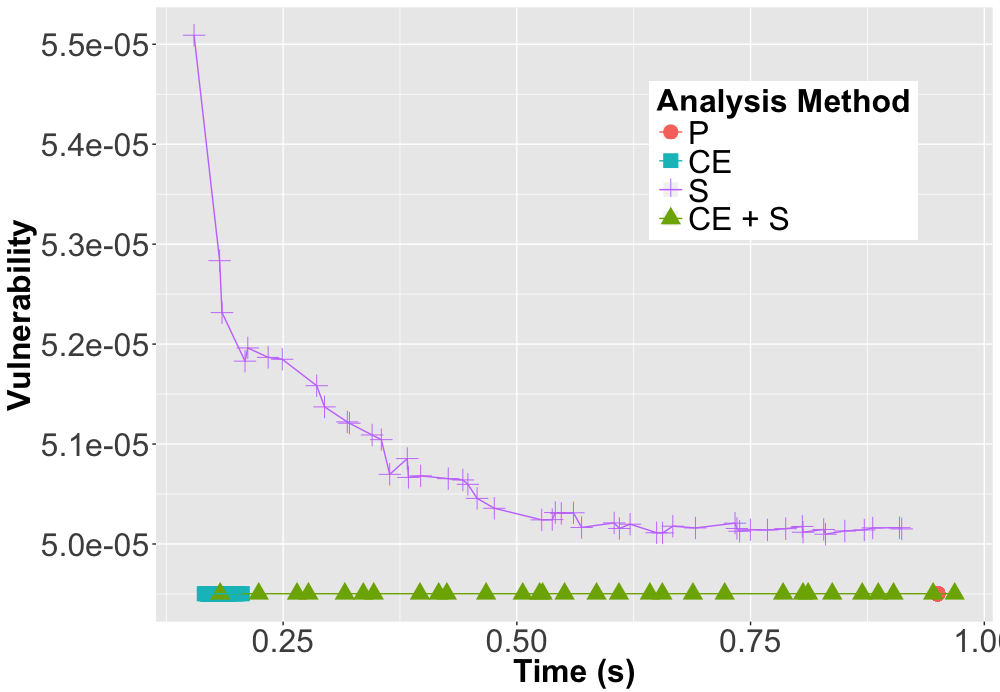}
  \end{minipage} &
  \begin{minipage}{.5\columnwidth}
    \includegraphics[width=\columnwidth]{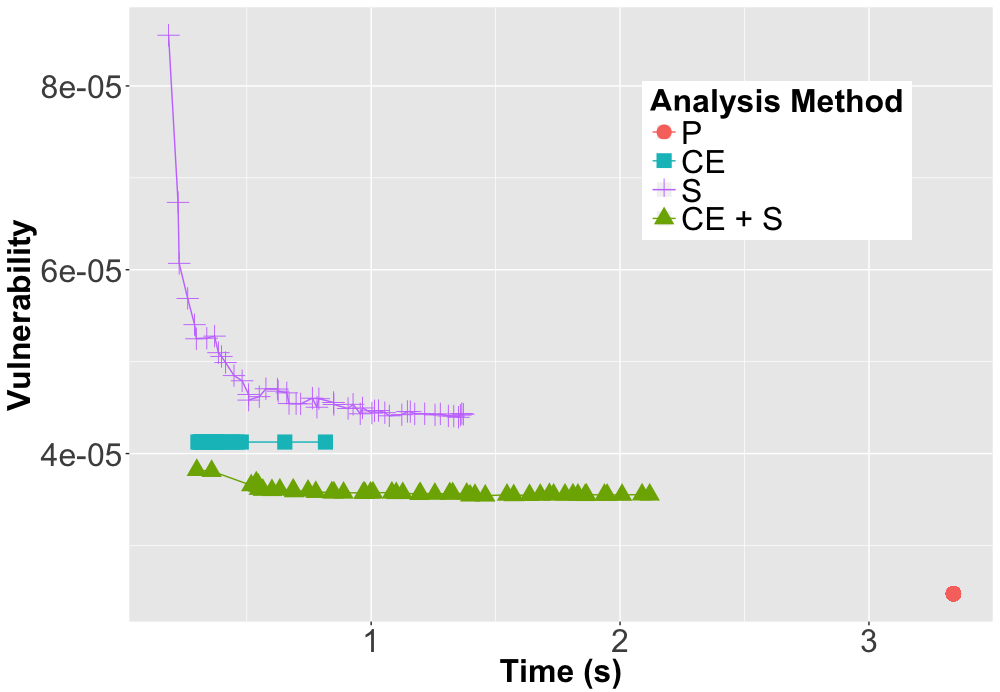}
  \end{minipage} \\
\\
  \parbox{.5\columnwidth}{\centering (a) Vulnerability vs. time,\\ $c = 1$ and $p
  = 1$} &
  \parbox{.5\columnwidth}{\centering (b) Vulnerability vs. time,\\ $c =
  2$ and $p = 4$} \\
\\
  \begin{minipage}{.5\columnwidth}
  \includegraphics[width=\columnwidth]{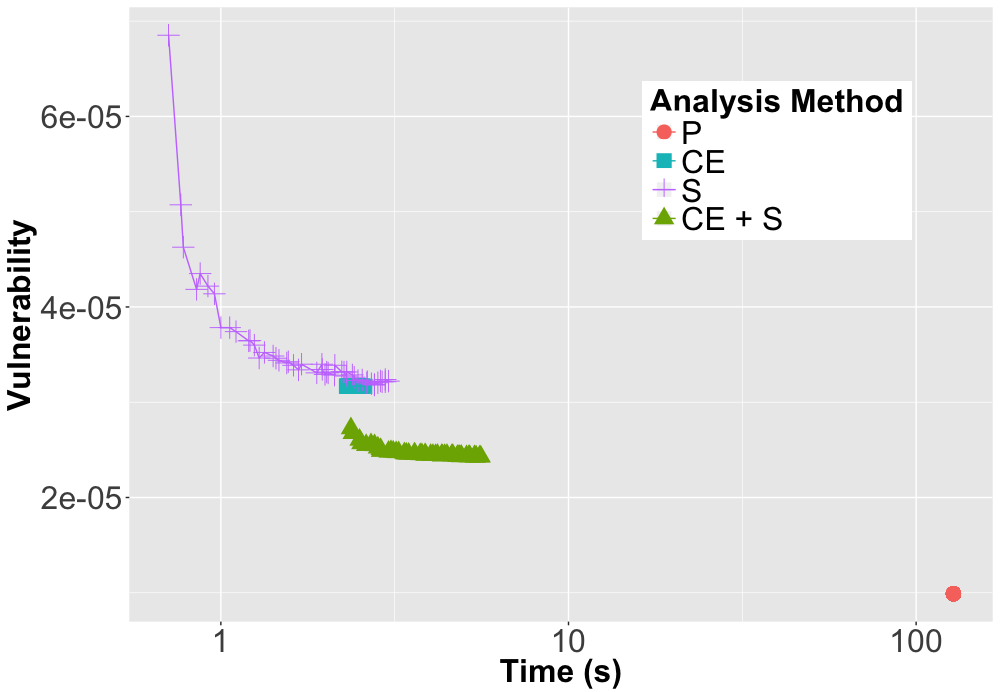}
  \end{minipage} &
  \begin{minipage}{.5\columnwidth}
    \includegraphics[width=\columnwidth]{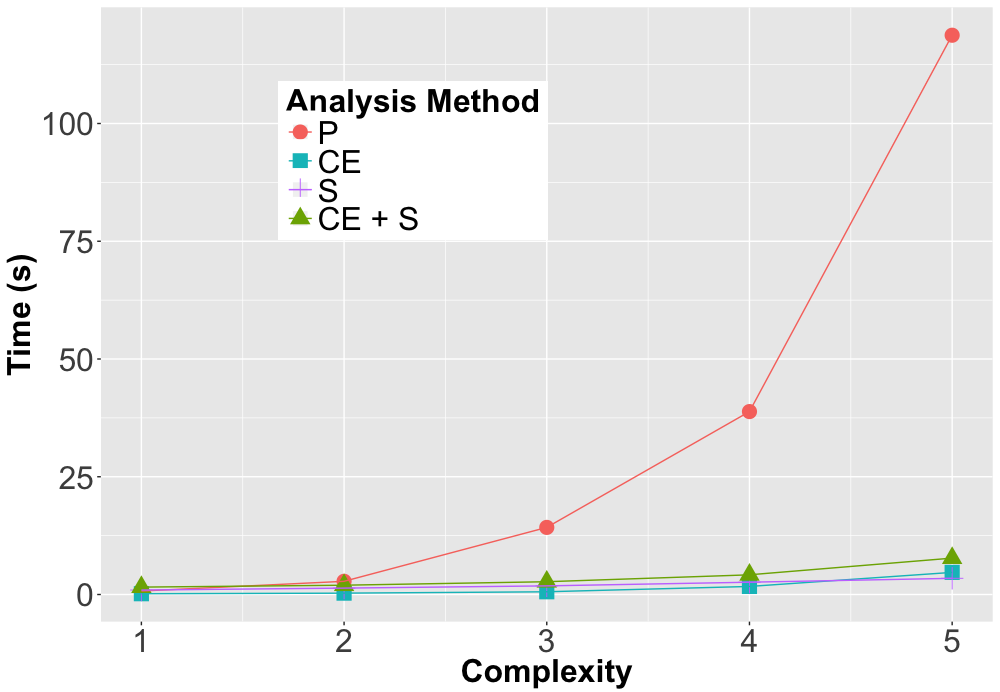}
  \end{minipage} \\
\\
  \parbox{.5\columnwidth}{\centering (c) Vulnerability vs. time,\\ $c =
  5$ and $p = 32$ (X-axis is log-scale)} &
  \parbox{.5\columnwidth}{\centering (d) Time vs. complexity,\\ $n =
                                          50,000$ and $p = 64$} \\
\end{tabular}
\caption{Experimental results}
\label{fig:results}
\end{figure}

In Figure~\ref{fig:results}(a) we fix $c = 1$ and $p = 1$.
In this configuration, baseline analysis P can compute the true vulnerability
in $\sim\!0.95$ seconds. Analysis CE is also able to compute the true vulnerability,
but in $\sim\!0.19$ seconds. Analysis S is able to compute a vulnerability to within $\sim\!5 \cdot e^{-6}$
of optimal in $\sim\!0.15$ seconds. These data points support two key observations. First, even a very modest
number of samples improves vulnerability significantly over just
analyzing with intervals. Second, concolic execution is
only slightly slower and can achieve the optimal vulnerability. Of course, concolic execution is not a panacea. As
we will see, a feature of this configuration is that no joins take place during abstract interpretation. This is
critical to the precision of the concolic execution.

In Figure~\ref{fig:results}(b) we fix $c = 2$ and $p = 4$. In contrast the the configuration
of Figure~\ref{fig:results}(a), the values for $c$ and $p$ in this configuration are not sufficient
to prevent all joins during abstract interpretaion. This has
the effect of taking polygons that represent individual paths through the program and joining them into
a single polygon representing many paths. We can see that this is the
case because baseline analysis P is
now achieving a better vulnerability than CE. However, one pattern from the previous configuration
persists: all three refinement methods (CE, S, CE+S) can achieve vulnerability within $\sim\!1 \cdot e^{-5}$ of P, but in $\frac{1}{4}$
the time. In contrast to the previous configuration,
analysis CE+S
is now able to make a modest improvement over CE (since it does not achieve the optimal).

In Figure~\ref{fig:results}(c) we fix $c = 5$ and $p = 32$. This configuration magnifies
the effects we saw in Figure~\ref{fig:results}(b). Similarly, in this configuration
there are joins happening, but the query is much more complex and the analysis
is much more precise. In this figure, we label the X axis as
a log scale over time. This is because analysis P took over two minutes to complete, in
contrast the longest-running refinement method, which took less than $6$ seconds. The relationship between the
refinement analyses is similar to the previous configuration. The key observation here is that, again, all
three refinement analyses achieve within $\sim\!3 \cdot e^{-5}$ of P, but this time in $4\%$ of the time
(as opposed to $\frac{1}{4}$ in the previous configuration).

Figure~\ref{fig:results}(d) makes more explicit the
relationship between refinements (CE, S, CE+S) and P. We fix $n = 50,000$ (the maximum) here, and
$p = 64$ (the maximum). We can see that as query complexity goes up, P
gets exponentially slower, while CE, S, and CE+S slow at a much
lower rate, while retaining (per the previous graphs) similar precision.

\subsection{Evacuation Problem}
\label{sec:evac}

\begin{table}[t]
\normalsize
\centering
  \caption{Analyzing a 3-ship resource allocation run}
  \label{tab:resource_alloc}
  \begin{tabular}{|l|c|c|}
    \hline
    \multicolumn{3}{|c|}{Resource Allocation (3 ships)} \\
    \hline
    Analysis & Time (s) & Vulnerability \\ \hline
    P & Timeout (5 min) & N/A \\
    I & $0.516$ & $1$ \\
    CE & $16.650$ & $1.997 \cdot 10^{-24}$ \\
    S & $1.487$ & $1.962 \cdot 10^{-24}$ \\
    CE+S & $17.452$ & $1.037 \cdot 10^{-24}$ \\ \hline
  \end{tabular}
\end{table}

We conclude this section by briefly discussing an analysis of an
execution of the resource allocation algorithm of
Figure~\ref{fig:bin-search}. In our experiment, we set the number of
ships to be three, where two were in range $d=300$ of the evacuation
site, and their sum-total berths ($500$) were sufficient to satisfy demand at the
site (also $500$). For our analysis refinements we set $n=1000$. Running the
algorithm, a total of seven pairs of \textit{Nearby}
and \textit{Capacity} queries were issued. In the end, the algorithm
selects two ships to handle the evacuation.

Table~\ref{tab:resource_alloc} shows the time to execute the algorithm
using the different analysis methods, along with the computed
vulnerability---this latter number represents the coordinator's view
of the most likely nine-tuple of the private data of the three ships involved
(x coordinate, y coordinate, and capacity for each). We can see that,
as expected, our refinement analyses are far more efficient than
baseline P, and far more precise than baseline I\@. The CE methods
are precise but slower than S. This is because of the need to count
the number of points in the DNF of the concolic path conditions,
which is expensive.

\section{Related Work}
\label{sec:related}

\textit{Quantifying Information Flow.}
%
There is a rich research literature on techniques
that aim to \emph{quantify} information that a program
may release, or has released, and then use that quantification as a
basis for policy.
One question is what measure of information release should be
used. Past work largely considers information theoretic measures, including
\emph{Bayes vulnerability}~\cite{smith09foundations} and \emph{Bayes risk}~\cite{Chatzikokolakis:08:JCS}, 
\emph{Shannon entropy}~\cite{shannon48communication},
and 
\emph{guessing entropy}~\cite{massey94guessing}.
The \emph{$g$-vulnerability} framework~\cite{alvim12gain} was recently
introduced to express measures having richer operational
interpretations, and subsumes other measures. 

Our work focuses on Bayes Vulnerability, which is related to min entropy.
Vulnerability is appealing operationally: As Smith~\cite{smith09foundations}
explains, it estimates the risk of the secret being guessed in one try. While
challenging to compute, this approach provides meaningful results for
non-uniform priors. Work that has focused on other, easier-to-compute metrics,
such as Shannon entropy and channel capacity, require deterministic
programs and priors that conform to
uniform distributions~\cite{McCamantE2008,backes09automatic,Mu:2009:inverval-qif,kopf:rybalchenko,KLEBANOV2014124,kr13}.
Like Mardziel et al.~\cite{mardziel13belieflong}, we are able to
compute the worst-case vulnerability, i.e., due to a particular
output, rather than
a \emph{static} estimate, i.e., as an expectation over all possible outputs.
K{\"o}pf and Basin \cite{koepfbasin07} originally proposed this idea, and
Mardziel et al. were the first to implement it, followed by several
others~\cite{Besson2014,Guarnieri17,Kucera:2017:SPP:3133956.3134079}. 

K\"opf and Rybalchenko~\cite{kopf:rybalchenko} (KR) also use sampling
and concolic execution to statically quantify information leakage. But
their approach is quite different from ours. KR uses sampling of a
query's inputs in lieu of considering (as we do) all possible outputs,
and uses concolic execution with each sample to ultimately compute
Shannon entropy, by underapproximation, within a confidence
interval. This approach benefits from not having to enumerate outputs,
but also requires expensive model counting \emph{for each sample}. By
contrast, we use sampling and concolic execution \emph{from the posterior}
computed by abstract interpretation, using the results to boost the
lower bound on the size/probability mass of the abstraction.
Our use of sampling is especially efficient, and the use of concolic
execution is completely sound (i.e., it retains 100\% confidence in
the result).  As with the above
work, KR requires deterministic programs and uniform priors.


\textit{Probabilistic Programming Langauges.}
A probabilistic program is essentially a lifting of a normal program operating
on single values to a program operating on distributions of values. As a
result, the program represents a joint distribution over its
variables~\cite{Gordon:2014:PP:2593882.2593900}. 
As discussed in this paper, quantifying the information released by a
query can be done by writing the query in a probabilistic programming
language (PPL) and representing the uncertain secret inputs as distributions. Quantifying
release generally 
%
corresponds to either the maximum likelihood estimation (MLE) problem
or the maximum a-posteriori probability (MAP) problem.
Not all PPLs support computation of MLE and
MAP, but several do. 

PPLs based on partial sampling
\cite{goodman08church, park08sampling} or full enumeration
\cite{radul07probscheme} of the state space are unsuitable in our
setting: they are either too inefficient or too imprecise. 
%
%
PPLs based on algebraic decision diagrams
\cite{claret12bayesian}, graphical models \cite{milch05blog}, and factor graphs
\cite{borgstrom11measure, pfeffer07ibal, InferNET14} translate programs into convenient
structures and take advantage of efficient algorithms for their manipulation or
inference, in some cases supporting MAP or MLE queries
(e.g.~\cite{figaro,Narayanan2016}). PSI~\cite{Gehr2016} supports exact
inference via computation of precise 
symbolic representations of posterior distributions, and has been used
for dynamic policy
enforcement~\cite{Kucera:2017:SPP:3133956.3134079}. Guarnieri et
al.~\cite{Guarnieri17} use probabilistic logic programming as the
basis for inference; it scales well but only for a class of queries
with certain structural limits, and which do not involve numeric
relationships. 

Our implementation for probabilistic computation and inference differs from
the above work in two main ways. Firstly, we are capable of
\emph{sound} approximation and hence can trade off precision for
performance, while maintaining soundness 
in terms of a strong security policy. Even when using sampling, we are
able to provide precise confidence measures. The second difference is our
\emph{compositional} representation of probability distributions, which is based on numerical
abstractions: intervals~\cite{cousot76static}, octagons~\cite{mine01octagon},
and polyhedra~\cite{Cousot:1978:ADL:512760.512770}. The posterior can
be easily used as the prior for the next query, whereas prior work
would have to repeatedly analyze the composition of past queries. 

A few other works have also focused on abstract interpretation, or
related techniques, for reasoning about probabilistic
programs. Monniaux~\cite{Monniaux_these} defines an abstract domain
for distributions. Smith~\cite{smith08probabilistic} describes
probabilistic abstract interpretation for verification of quantitative
program properties. Cousot~\cite{cousot12probabilistic} unifies these
and other probabilistic program analysis tools. However, these do not
deal with sound distribution conditioning, which is crucial for
belief-based information flow analysis.  Work by Sankaranarayanan et
al~\cite{Sankaranarayanan:2013:SAP:2491956.2462179} uses a combination
of techniques from program analysis to reason about distributions
(including abstract interpretation), but the representation does not
support efficient retrieval of the maximal probability, needed to
compute vulnerability.

\section{Conclusions}

Quantitative information flow is concerned with measuring the
knowledge about secret data that is gained by observing the answer to
a query. This paper has presented a combination of static
analysis using probabilistic abstract interpretation, sampling, and
underapproximation via concolic execution to compute high-confidence
upper bounds on information flow more
precisely and efficiently than past work. Experimental results show
dramatic improvements in overall precision and/or performance compared
to abstract interpretation alone. As next steps, we plan to integrate static
analysis and sampling more closely so as to avoid precision loss at
decision points in programs. We also look to extend programs to be
able to store random choices in variables, to thereby implement more
advanced probabilistic structures.

% ...


\bibliographystyle{splncs03}
\bibliography{bib,post2017,references}

All links were last followed on November 24, 2017.

\appendix

\section{Formal semantics and proofs}
\label{sec:dist-semantics}

Here we defined the probabilistic semantics for the programming
language given in Figure~\ref{fig:syntax}. The semantics of statement
$\stmt$, written $\pevalp{\stmt}{}$, is a function of the form
$\dists \rightarrow \dists$, i.e., it is a function from distributions
of states to distributions of states. We write
$\pevalp{\stmt}{\delta} = \delta'$ to say that the semantics of
$\stmt$ maps input distribution $\delta$ to output distribution
$\delta'$.

Figure~\ref{fig:dist-semantics} gives this denotational semantics along with
definitions of relevant auxiliary operations.\footnote{The notation
  $\sum_{x \given \phi} \rho$ can be read \emph{$\rho$ is the sum over
    all $x$ such that formula $\phi$ is satisfied} (where $x$ is bound
  in $\rho$ and $\phi$).} We write $\pevalp{E}{\sigma}$ to denote the
(integer) result of evaluating expression $E$ in $\sigma$, and
$\pevalp{B}{\sigma}$ to denote the truth or falsehood of $B$ in
$\sigma$. The variables of a state $\sigma$, written $ \dom{\sigma} $,
is defined by $ \dom{\sigma} $; sometimes we will refer to this set as
just the \emph{domain} of $ \sigma $. We will also use the this
notation for distributions; $ \dom{\delta} \defeq
\dom{\dom{\delta}}$. We write $\lfp$ as the least fixed-point
operator.

This semantics is standard. See Clarkson et
al.~\cite{clarkson09quantifying} or Mardziel et
al~\cite{mardziel13belieflong} for detailed explanations.





\begin{figure}
\centering
\begin{displaymath}
\begin{array}{rcl}
\pevalp{\sskip}{\delta} & = & \delta \\
\pevalp{\sassign{x}{\aexp}}{\delta} & = & \delta \bparen{x \ra \aexp} \\
\pevalp{\sif{B}{\stmt_1}{\stmt_2}}{\delta} & = &
\pevalp{\stmt_1}{(\dcond{\delta}{B})} + \pevalp{\stmt_2}{(\dcond{\delta}{\neg B})} \\
\evalp{\spif{q}{\stmt_1}{\stmt_2}}{\delta} & = &
\evalp{\stmt_1}{(q \cdot\delta)} + \evalp{\stmt_2}{((1-q) \cdot \delta)} \\
\pevalp{\sseq{\stmt_1}{\stmt_2}}{\delta} & = & \pevalp{\stmt_2}{\paren{\evalp{\stmt_1}{\delta}}} \\
\pevalp{\swhile{\bexp}{\stmt}}{} & = & \lfp\left[\lambda
f :\ \dists
\rightarrow \dists \lsep \lambda \delta \lsep \right. \\
& & \left. \quad f\paren{\pevalp{\stmt}{(\dcond{\delta}{B})}} +
       \paren{\dcond{\delta}{\neg B}}\right]
\end{array}
\end{displaymath}
where
\begin{displaymath}
\begin{array}{l@{\;\defeq\;}l}
\delta \bparen{x \ra \aexp} & \lambda \sigma \lsep \sum_{\tau \given \tau
  \bparen{x \ra \eeval{\aexp}{\tau}} = \sigma} \delta (\tau) \\
\delta_1 + \delta_2 & \lambda \sigma \lsep \delta_1(\sigma) +
\delta_2(\sigma) \\
\dcond{\delta}{\bexp} & \lambda \sigma \lsep \aif \eeval{\bexp}{\sigma} \athen
\delta(\sigma) \aelse 0 \\
p \cdot \delta & \lambda \sigma \lsep p \cdot \delta(\sigma) \\
\pmass{\delta} & \sum_\sigma \delta(\sigma) \\
\normal{\delta} & \frac{1}{\pmass{\delta}} \cdot \delta \\
\drevise{\delta}{B} & \normal{\dcond{\delta}{B}} \\
\delta_1 \times \delta_2 & \lambda(\sigma_1, \sigma_2) \lsep
\delta_1(\sigma_1) \cdot \delta_2(\sigma_2) \\
\dot{\sigma} & \lambda \sigma_0 \lsep \aif \sigma = \sigma_0 \athen 1 \aelse
0 \\
\project{\sigma}{V} & \lambda x \in \vars_V \lsep \sigma(x)\\
\project{\delta}{V} & \lambda \sigma_V \in \states_V \lsep
\sum_{\tau \given \project{\tau}{V} = \sigma_V} \delta(\tau) \\
\forget{x}{\delta} & \project{\delta}{\paren{\dom{\delta} - \set{x}}} \\
\nzset{\delta} & \{\sigma \given \delta(\sigma) > 0\}
\end{array}
\end{displaymath}
\vspace*{-.1in}
\caption{Distribution semantics}
\label{fig:dist-semantics}
\end{figure}

\subsection{Proofs}
\label{sec:proofs}

Here we restate the soundness theorems for our techniques, and include
their proofs.

\setcounter{theorem}{1}

\begin{theorem}[Sampling is Sound] \hfill \\
  \label{sampling_proof}
  If $\delta_0 \in \ppconc{\ppoly_0}$, $\abspevalp{\stmt}{\ppoly_0} =
  \ppoly$, and $\pevalp{\stmt}{\delta_0} = \delta$ then
  \[\delta_{T} \in \ppconc{\ppoly_{T+}} \text{ with confidence }\omega\]
  where
  \begin{align*}
  \delta_{T} &\defeq \project{\dcond{\delta}{(r = o)}}{T} \\
  \ppoly_{T} &\defeq \project{\dcond{\ppoly}{(r = o)}}{T} \\
  \ppoly_{T+} &\defeq \ppoly_{T} \text{ sampling revised with confidence }\omega.
  \end{align*}
  %
\end{theorem}
\begin{proof}
Suppose we have some $\delta_0 \in \ppconc{\ppoly_0}$ whereby
$\pevalp{\stmt}{\delta_0} = \delta$. We want to prove that $\delta_T \in
\ppconc{\ppoly_{T+}}$. Per Definition~\ref{def:ppoly}, this means we must show that
\begin{enumerate}
\item[(1)] $\nzset{\delta_T} \subseteq \ppconc{\getpoly{T+}}$
\item[(2)] $\smin{T+} \leq |\nzset{\delta_T}| \leq \smax{T+}$
\item[(3)] $\mmin{T+} \leq \pmass{\delta_T} \leq \mmax{T+}$
\item[(4)] $\forall \sigma \in \nzset{\delta_T}.\, \pmin{T+} \leq \delta_T(\sigma) \leq \pmax{T+}$
\end{enumerate}
Our proof goes as follows. First, we know that $\delta_T \in
\ppconc{\getpoly{T}}$ by  Theorem~\ref{soundness}, Lemma 15 and Lemma 7 of Mardziel et al. By
  Definition~\ref{def:ppoly}, this means
\begin{enumerate}
\item[(a)] $\nzset{\delta_T} \subseteq \ppconc{\getpoly{T}}$
\item[(b)] $\smin{T} \leq |\nzset{\delta_T}| \leq \smax{T}$
\item[(c)] $\mmin{T} \leq \pmass{\delta_T} \leq \mmax{T}$
\item[(d)] $\forall \sigma \in \nzset{\delta_T}.\, \pmin{T} \leq \delta_T(\sigma) \leq \pmax{T}$
\end{enumerate}
So (1) and (4) follow directly from (a) and (d), since $\pmin{T}
= \pmin{T+}$, $\pmax{T} = \pmax{T+}$, and $\getpoly{T} = \getpoly{T+}.$

To prove (2), we argue as follows. Let $p =
\frac{|\nzset{\delta_T}|}{\psize{\getpoly{T}}}$, which
    represents the probability that a randomly selected
point from $\getpoly{T}$ is in $\nzset{\delta_T}$. From the computed credible
interval over the Beta distribution, we have that $p \in [p_L, p_U ]$ with
confidence $\omega$. As such,
\[
\begin{array}{rcl}
    p_L &\leq p \leq & p_U\\
    p_L &\leq \frac{|\nzset{\delta_{T}}|}{\psize{\getpoly{T}}} \leq & p_U \\
    p_L \cdot \psize{\getpoly{T}} &\leq |\nzset{\delta_{T}}| \leq & p_U \cdot \psize{\getpoly{T}} \\
    \smin{T+} & \leq |\nzset{\delta_{T}}| \leq & \smax{T+} \\
  \end{array}
\]
which is the desired result.

To prove (3), first consider that if $\mmin{T+} = \mmin{T}$ then the
first half of (3) follows from the first half of (c). Otherwise, we
have that $\mmin{T+} = \pmin{T} \cdot \smin{T+}$. Then we can reason
the first half of (3) holds using the following reasoning:
\[
\begin{array}{rcll}
  \smin{T+} & \leq & |\nzset{\delta_T}| & \text{by (2)} \\
  \pmin{T+} \cdot \smin{T+} & \leq & \pmin{T+} \cdot
                                     |\nzset{\delta_T}| & \text{as }\pmin{T+}\text{ nonneg.}\\
  \mmin{T+} & \leq & \pmin{T+} \cdot
                                     |\nzset{\delta_T}| & \text{by def.}\\
            & & = \Sigma_{\sigma \in
  \nzset{\delta_T}} \pmin{T+} & \\
  &  & \leq
                                                          \Sigma_{\sigma
                                                          \in
                                                          \nzset{\delta_T}}
                                                          \delta_T(\sigma)
                                        & \text{by (4)} \\
& & = \pmass{\delta_T} & \\
  \end{array}
\]
We can prove the soundness of $\mmax{T+}$ (the other half of (3)) with
similar reasoning.
\end{proof}

\begin{theorem}[Concolic Execution is Sound] \hfill \\
  \label{concolic_proof}
  If $\delta_0 \in \ppconc{\ppoly_0}$, $\abspevalp{\stmt}{\ppoly_0} = \ppoly$, and $\pevalp{\stmt}{\delta_0} = \delta$ then

  \[ \delta_{T} \in \ppconc{\ppoly_{T+}} \]

  where
  \begin{align*}
  \delta_{T} &\defeq \project{\dcond{\delta}{(r = o)}}{T} \\
  \ppoly_{T} &\defeq \project{\dcond{\ppoly}{(r = o)}}{T} \\
  \ppoly_{T+} &\defeq \ppoly_{T} \text{ concolically revised.}
  \end{align*}
\end{theorem}
\begin{proof}
  Our proof is quite similar to that of
  Theorem~\ref{sampling_proof}. Once again we proceed to show the four
  elements of Definition~\ref{def:ppoly}, where (1) and (4) hold by
  construction. To prove (2) we are only concerned with the inequality
  $\smin{T} \leq |\nzset{\delta_{T}}|$, since $\smax{T+} =
  \smax{T}$. We know by the soundness of $\ppoly_{T}$ that
\[
    \smin{T} \leq |\nzset{\delta_{T}}|
\]
  From the definition of concolic
 execution we have $\{ \sigma \mid \sigma \in C_T \land \sigma \models
 \pi \} \subseteq \{ \sigma \mid \delta_{T}(\sigma) > 0 \}$. Notice that
 this is just saying that the concolic execution is a valid
 under-approximation for the support.
 \mwh{Is
   the previous really by the definition? It seems like something we
   should have to prove.} From this we know that:
  \begin{align*}
    |\{ \sigma \mid \sigma \in C_T \land \sigma \models \pi \}| &\leq |\{ \sigma \mid \delta_{T}(\sigma) > 0 \}|\\
    \psize{C_T \sqcap (\bigsqcup_i C_i)} &\leq |\{ \sigma \mid \delta_{T}(\sigma) > 0 \}| \\
    \psize{C_T \sqcap (\bigsqcup_i C_i)} &\leq |\nzset{\delta_{T}}| \\
    \smin{T+} &\leq |\nzset{\delta_{T}}| \\
  \end{align*}
Given (2), our proof of (3) proceeds similarly to Theorem~\ref{sampling_proof}.
\end{proof}

\begin{theorem}[Concolic and Sampling Composition is Sound] \hfill \\
  \label{combo_proof}
  If $\delta_0 \in \ppconc{\ppoly_0}$, $\abspevalp{\stmt}{\ppoly_0} = \ppoly$, and $\pevalp{\stmt}{\delta_0} = \delta$ then

  \[ \delta_{T} \in \ppconc{\ppoly_{T+}} \]

  where
  \begin{align*}
  \delta_{T} &\defeq \project{\dcond{\delta}{(r = o)}}{T} \\
  \ppoly_{T} &\defeq \project{\dcond{\ppoly}{(r = o)}}{T} \\
  \ppoly_{T+} &\defeq \ppoly_{T} \text{ sampling-and-concolically revised with confidence }\omega
  \end{align*}
\end{theorem}
\begin{proof}
Our proof is quite similar to that of Theorem 2. Once again we proceed
to show the four elements of Definition 2, where (1) and (4) hold by construction.
To prove (2), consider the set $|\nzset{\delta_T}|$. For any state $\sigma_T \in \nzset{\delta_T}$
it must be the case that either $\sigma_T \in \getpoly{}$ or $\sigma_T \in \getpoly{T} \setminus \getpoly{}$.
This is because $\getpoly{} \subseteq \getpoly{T}$. Thus, we have

\[
|\nzset{\delta_T}| = |\nzset{\getpoly{T} \setminus \getpoly{}}| + |\nzset{\getpoly{}}|
\]

Additionally, since $\getpoly{}$ represents a sound under-approximation of the support, we know that

\[
|\nzset{\getpoly{}}| = \psize{\getpoly{}}
\]

So, $p = \frac{|\nzset{\getpoly{T} \setminus \getpoly{}}|}{\psize{\getpoly{T} \setminus \getpoly{}}}$ represents
the probability that a point in the region surrounding $\getpoly{}$ is in the support of $\delta_T$. Note that our
procedure implements a uniform, random sample over this region. Thus, from the computed credible interval over the
Beta distribution, we have that $p \in [p_L, p_U]$ with confidence $\omega$. As such, \mwh{The following formats badly}
\[
\begin{array}{rcl}
    p_L &\leq p \leq & p_U\\
    p_L &\leq \frac{|\nzset{\getpoly{T} \setminus \getpoly{}}|}{\psize{\getpoly{T} \setminus \getpoly{}}} \leq & p_U \\
    p_L \cdot \psize{\getpoly{T} \setminus \getpoly{}} &\leq |\nzset{\getpoly{T} \setminus \getpoly{}}| \leq & p_U \cdot \psize{\getpoly{T} \setminus \getpoly{}} \\
    p_L \cdot \psize{\getpoly{T} \setminus \getpoly{}} + |\nzset{\getpoly{}}| &\leq |\nzset{\getpoly{T} \setminus \getpoly{}}| + |\nzset{\getpoly{}}| \leq & p_U \cdot \psize{\getpoly{T} \setminus \getpoly{}} + |\nzset{\getpoly{}}| \\
    p_L \cdot \psize{\getpoly{T} \setminus \getpoly{}} + |\nzset{\getpoly{}}| &\leq |\nzset{\delta_T}| \leq & p_U \cdot \psize{\getpoly{T} \setminus \getpoly{}} + |\nzset{\getpoly{}}| \\
    p_L \cdot \psize{\getpoly{T} \setminus \getpoly{}} + \psize{\getpoly{}} &\leq |\nzset{\delta_T}| \leq & p_U \cdot \psize{\getpoly{T} \setminus \getpoly{}} + \psize{\getpoly{}} \\
    p_L \cdot (\psize{\getpoly{T}} - \psize{\getpoly{}}) + \psize{\getpoly{}} &\leq |\nzset{\delta_T}| \leq & p_U \cdot (\psize{\getpoly{T}} - \psize{\getpoly{}}) + \psize{\getpoly{}} \\
    \smin{T+} & \leq |\nzset{\delta_{T}}| \leq & \smax{T+} \\
  \end{array}
\]

which is the desired result. \\

Given (2), our proof of (3) proceeds similarly to Theorem~\ref{sampling_proof} and~\ref{concolic_proof}.
\end{proof}

\section{Query code}
\label{app:code}

The following is the query code of the example developed in
Section~\ref{sec:samp_and_conc}. Here, \lstinline|s_x| and
\lstinline|s_y| represent a ship's secret location. The
variables \lstinline|l1_x|, \lstinline|l1_y|, \lstinline|l2_x|,
\lstinline|l2_y|, and \lstinline|d| are inputs to the query. The
first pair represents position $L_1$, the second pair represents the
position $L_2$, and the last is the distance threshold, set to $4$. We
assume for the example that $L_1$ and $L_2$ have the same y
coordinate, and their x coordinates differ by 6 units.

We express the query in the language of Figure~\ref{fig:syntax} basically as
follows:
\begin{verbatim}
  d_l1 := |s_x - l1_x| + |s_y - l1_y|;
  d_l2 := |s_x - l2_x| + |s_y - l2_y|;
  if (d_l1 <= d || d_l2 <= d) then
    out := true // assume this result
  else
    out := false
\end{verbatim}
The variable \lstinline|out| is the result of the query. We simplify
the code by assuming the absolute value function is built-in; we can
implement this with a simple conditional. We run this query
probabilistically under the assumption that \lstinline|s_x| and
\lstinline|s_y| are uniformly distributed within the range given in
Figure~\ref{fig:data}. We then condition the output on the assumption that
\lstinline|out = true|. When using intervals as the baseline of
probabilistic polyhedra, this produces the result given in the upper
right of Figure~\ref{fig:prob-AI}(b); when using convex polyhedra, the
result is shown in the lower right of the figure. The use of sampling
and concolic execution to augment the former is shown via arrows
between the two.

\end{document}